\documentclass[11pt]{article}
\usepackage{fullpage}
\usepackage{amsmath,amsfonts,amsthm,amssymb}
\usepackage[colorlinks=true, linkcolor=blue, urlcolor=blue, citecolor=gray]{hyperref}
\usepackage{url}
\usepackage{algorithm}
\usepackage{algorithmic}
\usepackage{bbm}
\usepackage{float}
\usepackage{framed}
\usepackage{enumerate}
\usepackage{color}
\usepackage[usenames,dvipsnames,svgnames,table]{xcolor}
\usepackage{array}
\usepackage[para]{footmisc}
\usepackage{epsfig}
\usepackage{framed}
\usepackage[framemethod=tikz]{mdframed}
\usepackage{caption}
\usepackage[most]{tcolorbox}
\usepackage{wrapfig}

\newcommand{\cC}{\mathcal{C}}
\newcommand{\overbar}[1]{\mkern 1.5mu\overline{\mkern-1.5mu#1\mkern-1.5mu}\mkern 1.5mu}
\newcommand{\freq}{\mathop\mathrm{freq}}
\DeclareMathOperator{\pot}{pot}
\newcommand{\poly}{\mathop\mathrm{poly}}

\newcommand{\POT}{\mathop\mathrm{POT}}
\newcommand{\AND}{\mathop\mathrm{AND}}
\newcommand{\OR}{\mathop\mathrm{OR}}
\newcommand{\dec}{\mathop\mathrm{dec}}

\newcommand{\cN}{\mathcal N}
\newcommand{\cA}{\mathcal A}
\newcommand{\PSNNdet}{\mathcal{SNN}^{\mathrm{poly}}_{\mathrm{det}}}
\newcommand{\PSNNrand}{\mathcal{SNN}^{\mathrm{poly}}_{\mathrm{rand}}}
\newcommand{\SNNdet}{\mathcal{SNN}_{\mathrm{det}}}
\newcommand{\SNNrand}{\mathcal{SNN}_{\mathrm{rand}}}
\newcommand{\Streamdet}{\mathcal{ST}_{\mathrm{det}}}
\newcommand{\Streamrand}{\mathcal{ST}_{\mathrm{rand}}}
\newcommand{\PStreamdet}{\mathcal{ST}^{\mathrm{poly}}_{\mathrm{det}}}
\newcommand{\PStreamrand}{\mathcal{ST}^{\mathrm{poly}}_{\mathrm{rand}}}

\newcommand{\inc}{\mathrm{inc}}
\newcommand{\countt}{\mathrm{count}}
\newcommand{\B}{\{ 0,1 \}}
\newtheorem{theorem}{Theorem} 

\newtheorem{corollary}[theorem]{Corollary}
\newtheorem{proposition}[theorem]{Proposition}
\theoremstyle{plain}
\newtheorem{observation}[theorem]{Observation}
\newtheorem{question}[theorem]{Question}
\newtheorem{lemma}[theorem]{Lemma}
\newtheorem*{lemma*}{Lemma}
\theoremstyle{plain}
\newtheorem{fact}[theorem]{Fact}

\newtheorem{claim}[theorem]{Claim}
\newtheorem{definition}[theorem]{Definition}
\renewcommand{\paragraph}[1]{\vspace{0.15cm}\noindent {\bf #1}}
\title{Spiking Neural Networks Through the Lens of Streaming Algorithms}
\date{}
\author{
	Yael Hitron \\
	\small Weizmann Institute\\
	\small yael.hitron@weizmann.ac.il
	\and				
	Cameron Musco\\
	\small University of Massachusetts Amherst \\
	\small cmusco@cs.umass.edu
		\and				
	Merav Parter\\
	\small Weizmann Institute \\
	\small merav.parter@weizmann.ac.il \thanks{YH and MP are supported in part by the ISF-BFS grant 2017758.}
}
\begin{document}
\maketitle

\begin{abstract}
We initiate the study of biological neural networks from the perspective of streaming algorithms. Like computers, human brains suffer from memory limitations which pose a significant obstacle when processing large scale and dynamically changing data. In computer science, these challenges are captured by the well-known streaming model, which can be traced back to Munro and Paterson `78 and has had significant impact in theory and beyond. In the classical streaming setting, one must compute some function $f$ of a stream of updates $\mathcal{S} = \{u_1,\ldots,u_m\}$, given restricted single-pass access to the stream. The primary complexity measure is the space used by the algorithm. 

In contrast to the large body of work on streaming algorithms, relatively little is known about the computational aspects of data processing in biological neural networks. In this work, we seek to connect these two models, leveraging techniques developed in for streaming algorithms to better understand neural computation. In particular, we consider the spiking neural network model, a distributed model of biological networks in which
nodes (neurons) are connected by edges (synapses), and communicate with their neighbors via spiking (i.e., firing).
Our primary goal is to design networks for various computational tasks using as few auxiliary (non-input or output) neurons as possible. The number of auxiliary neurons can be thought of as the `space' required by the network. 

Previous algorithmic work in spiking neural networks has many similarities with streaming algorithms. However, the connection between these two space-limited models has not been formally addressed. We take  the first steps towards understanding this connection. On the upper bound side, we design neural algorithms based on known streaming algorithms for fundamental tasks, including distinct elements, approximate median, heavy hitters, and more. The number of neurons in our neural solutions almost matches the space bounds of the corresponding streaming algorithms. As a general algorithmic primitive, we show how to implement the important streaming technique of linear sketching efficient in spiking neural networks. On the lower bound side, we give a generic reduction, showing that any space-efficient spiking neural network can be simulated by a space-efficiently streaming algorithm. This reduction lets us translate 
streaming-space lower bounds into nearly matching neural-space lower bounds, establishing a close connection between these two models.
\end{abstract}
 
\newpage
\tableofcontents
\newpage
\clearpage

\section{Introduction}\vspace{-7pt}
In this work, we seek to understand the role of \emph{memory constraints} in neural data processing. 
We consider data-stream tasks, in which a long stream of inputs is presented over time and a neural network must evaluate some function $f$ of this stream. Examples include identifying frequent input patterns (items) or estimating summary statistics, such as the number of distinct items presented. The network cannot store the full stream and so must maintain some form of compressed representation in its working memory, which allows the eventual computation of $f$.
The primary objective is to compute $f$ with as few auxiliary (non-input or output) neurons as possible. The number of auxiliary neurons can be thought of as the `space' required by the network. 

\noindent In computer science, data processing under space limitations is extensively studied in the area of streaming algorithms \cite{munro1980selection,muthukrishnan2005data}.
We leverage this body of work to further our understanding of space-efficient neural networks. We start by designing neural networks for a large class of data-stream tasks, building off fundamental streaming algorithms and techniques, such as linear sketching. We also establish general connections between these models, showing that streaming-space lower bounds can be translated to neural-space lower bounds. We hope that these connections are a first step in extending work on streaming computation to better understand neural processing of massive and dynamically changing data under memory constraints.

\paragraph{The spiking neural network (SNN) model} \cite{maass1996computational,maass1997networks}.
A spiking network is represented by a directed weighted graph over $n$ input neurons, $r$ output neurons, and $s$ auxiliary neurons. The edges of the graph represent synapses of different strengths connecting the neurons. 
The network evolves in discrete, synchronous {rounds} as a Markov chain where each neuron $u$ acts as a (possibly probabilistic) threshold gate that either fires (spikes) or is silent in each round. In round $t$, 
the firing status of $u$ depends on the firing status of its incoming neighbors in the preceding round $t-1$, and the strength of the connections from these neighbors. In \emph{randomized} SNNs, there are possibly two sources of randomness: the spiking behavior of the neurons and the selection of random edge weights in the network. In \emph{deterministic} SNNs, the neurons are deterministic threshold gates and the edge weights are deterministically chosen. Aside from their relevance in modeling biological computation, SNNs have received significant attention as more energy efficient alternatives to traditional artificial neural networks \cite{lee2016training,tavanaei2019deep}.

A recent series of works in the emerging area of \emph{algorithmic SNNs} \cite{maass1997networks,maass2000computational,dasgupta2017neural,lynch2017spiking,lynch2017computational,Valiant17,chou2018algorithmic,Legenstein0PV18,su2019spike, 0001PVL19,PapadimitriouV19,HitronPP20} focuses on network design tasks. In this framework, given a target function $f:\{0,1\}^n \to \{0,1\}^r$, one seeks to design a space-efficient SNN (with few auxiliary neurons) that converges rapidly to an output spiking pattern matching $f(x)$ when the input spiking pattern matches $x$. Space-efficient SNNs have been devised for the winner-takes-all problem \cite{lynch2017computational,su2019spike}, similarity testing and compression \cite{LynchMP17,PapadimitriouV19}, clustering \cite{HitronLMP20,Legenstein0PV18}, approximate counting, and time estimation \cite{lynch2019integrating,HitronP19}.  Interestingly, many of these works borrow ideas from related streaming algorithms. However, despite the flow of ideas from streaming to neural algorithms, the connection between these models has not been studied formally.

\paragraph{The streaming model} \cite{munro1980selection,muthukrishnan2005data}. 
A data-stream is a sequence of updates $\mathcal{S}=\{u_1,\ldots,u_m\}$.
A streaming algorithm $\mathcal{A}$ computes some function of $\mathcal{S}$, given restricted access to the stream. In the standard single-pass model, the algorithm can only read the updates in $\mathcal S$ once, in the order they are presented. 

Most commonly, and throughout this work, each update $u_i$ represents the insertion or deletion of an item $x_i$ belonging to a universe $U$ with $|U| = n$. Without loss of generality, we will always consider $U$ to be the set of integers $[n] = 1,\ldots,n$, and $f$ is a function of the frequency vector $\overbar{z} \in \mathbb{Z}^n$, which tracks the total frequency of each item in the stream (the number of insertions minus the number of deletions). In the \emph{insertion-only} setting, only insertions are allowed -- i.e., each update increments some entry of $\overbar{z}$. In the general  \emph{turnstile} (dynamic) setting, there are both insertions and deletions -- i.e., increments and decrements to entries in $\overbar{z}$. The primary complexity measure of a streaming algorithm is the \emph{space} (measured in number of bits) required to maintain the evaluation of $f$ on the data-stream. 

\paragraph{Neural networks from a streaming perspective.}
Our primary goal is to devise space-efficient spiking neural networks that solve natural data-stream tasks, which mirror data processing tasks solved in real biological networks.
In light of the large collection of space-efficient streaming algorithms that have been designed for various problems, we start by asking:
\vspace{-5pt}
\begin{question}\label{ques:upper}
\emph{Is it possible to translate a space-efficient streaming algorithm for a given task into a space-efficient SNN algorithm for that task? Do generic reductions from SNNs to streaming exist?}
\end{question}
\vspace{-5pt}  
The streaming literature is also rich with space lower bounds. For many classical data-stream problems, these lower bounds are nearly tight. To obtain space lower bounds for SNNs, we ask if reductions in the reverse direction exist:
\vspace{-5pt}
\begin{question}\label{ques:lower}
\emph{Is it possible to translate a space-efficient SNN for a given task into a space-efficient streaming algorithm for that task?}
\end{question}\vspace{-5pt}
An affirmative answer to both of these questions would imply that the streaming and SNN models are, roughly speaking, computationally equivalent. A priori, it is unclear if this is the case. On the one hand, streaming algorithms have the potential to be more space-efficient than SNNs. For example, a space-efficient algorithm may still have a lengthy description, which is not taken into account in its space complexity. In the SNN setting, where the algorithm description and memory are both encoded by the auxiliary neurons in the network and their connections, a lengthy description may lead to a large, and hence not space-efficient network.

On the other hand, SNNs have the potential to be more space-efficient than streaming algorithms. For example, a randomized SNN with a large number of input neurons but a small number of auxiliary neurons may have a large number of random bits encoded in random connections between its inputs and auxiliary neurons. These bits are not counted as part of its space complexity. In contrast, a streaming algorithm that requires persistent access to many random bits must store these bits, possibly leading to large space complexity.
 
\vspace{-10pt}
\subsection{Our Results} \vspace{-5pt}
We take the first steps towards formally understanding the connections between streaming algorithms and spiking neural networks. The first part of the paper is devoted to studying upper bounds for SNNs, addressing Question \ref{ques:upper}. 
We design space-efficient neural networks for a wide class of streaming problems by simulating their respective streaming algorithms. These simulations must overcome several challenges in implementing traditional algorithms in neural networks. Most notably, in an SNN, the spiking status of the auxiliary neurons encodes the working memory of the algorithm, and their connections encode the algorithm itself. A space-efficient network with few auxiliary neurons thus inherently has limited ability to express complex algorithms. In many data-stream algorithms, the target space complexity is only polylogarithmic in the input size, making this challenge significant. Additionally, unlike traditional algorithms, a neural network evolves continuously in response to its inputs. This leads to synchronization issues -- for example, if an input is not presented for a sufficient number of rounds, the firing status of the network may not converge to a proper state before the next input is presented.

The second part of the paper focuses on lower bound aspects, addressing Question \ref{ques:lower}. We show that any space-efficient neural network can be translated into a space-efficient streaming algorithm, while paying a small additive term (logarithmic in the stream length/universe size). For deterministic SNNs, such a reduction is not difficult. For randomized SNNs, the reduction is more involved, as it must account for the large number of random bits that may be implicitly stored in the random edge weights of the network.
Throughout, we use the $\widetilde{O}()$ notation to hides factors that are poly-logarithmic in $n,m$ and $1/\delta$, where $n$ is the size of the domain, $m$ is a bound on the stream length and $\delta$ is the error parameter. 
\vspace{-10pt}
\subsubsection{Efficient Streaming Algorithms Yield Efficient SNNs} 
We consider data-stream tasks in which each update is an insertion or deletion of an integer item $x \in [n]$, and $f$ is a function of the frequency vector $\overbar{z} \in \mathbb{Z}^n$ of these items. In the streaming setting, each update can be thought as an $n$-length vector with a single $\pm 1$ entry, corresponding to an item insertion or deletion. In the SNN setting, each update may be encoded as the firing of one of $n$ input neurons along with a sign neuron indicating if the update is an increment or a decrement. Or, the update may be encoded via $O(\log n)$ input neurons, indicating the item to be inserted or deleted. These different encodings correspond to different natural settings -- the first corresponds to a network that collects firing statistics from a large set of inputs and the second to a network that records statistics on a large number of possible input patterns, encoded in the spiking patterns of a smaller number of input neurons.
 
In either case, each input is presented for some \emph{persistence time}, a certain number of rounds in which the input is fixed to allow the network state to converge before the next input is presented. 

\paragraph{Linear sketching.} A linear sketching algorithm is a streaming algorithm in which the state of the algorithm is a linear function of the updates seen so far. In particular, the state can be represented as the multiplication of a sketching matrix $A \in \mathbb{R}^{r \times n}$ with the frequency vector $\overbar{z} \in \mathbb{Z}^n$. Such algorithms have many useful properties applicable in both the turnstile setting and in distributed settings. For example, the additive nature of these algorithms allows one to split the data-stream across multiple sites, which can process the data in an independent manner. Additionally, the obliviousness of linear sketching algorithms to the ordering of the stream yields an efficient generic derandomization scheme using the Nisan's PRG for space bounded computation \cite{indyk2006stable}.  Linear sketching algorithms constitute the state-of-the-art algorithms for essentially all problems in the turnstile model, including heavy-hitters, coresets for clustering problems \cite{indyk2011k}, and $\ell_p$ estimation \cite{cormode2003comparing}.  In fact, Li, Nguyen and Woodruff \cite{LiNW14} present a general reduction from the streaming turnstile model to linear sketching. This reduction, and its caveats have been further studied in a recent work by Kallaugher and Price \cite{Kallaugher20}.
Given their ubiquity in turnstile streaming algorithms, an important step in designing space-efficient SNNs for data-stream problems is an efficient implementation of linear sketching in the neural setting. We give such an implementation:
\vspace{-3pt}
\begin{theorem}[Linear Sketch]\label{lem:linear-sketch}
Let $\cA$ be an algorithm approximating a function $f(\overbar{x})$
in the turnstile model using a linear sketch with an integer matrix $A$ of size $r \times n$. Let $\ell$ be a bound on the maximum entry in $|A\overbar{x}|$ for every item $\overbar{x}$. 
There exists a network $\cN$ with $n+1$ input neurons, $r \cdot (\lceil\log \ell \rceil+1)$ output neurons, $O(r \cdot \log \ell)$ auxiliary neurons which implements $\cA$ in the following sense. 
The first $n$ input neurons $x=(x_1, \ldots, x_n)$ represent the inserted item $[1,n]$, and the additional input neuron $s$ indicates the sign of the update. Each input update has a persistence time of $O(\log \ell)$ rounds.
The output neurons are divided into $r$ vectors $\overbar{y}_1, \ldots, \overbar{y}_r$ each of length $\log \ell$, and $r$ neurons $s_1, \ldots , s_r$. For every $i \in \{1,\ldots, r\}$, the decimal value of the binary vector 
$\overbar{y}_i$ is equal to the absolute value of the $i^{th}$ entry of $A \cdot \overbar{z}$, and the sign neuron $s_i$ indicates the sign, where $\overbar{z}$ is the sum of all input items presented in the current stream.
\end{theorem} \vspace{-5pt}
Theorem \ref{lem:linear-sketch} applies to linear sketches using integer matrices, which are commonly used, see \cite{LiNW14}. 
Via scaling, the construction can be extended to rational matrices as well.  
We note that the network of Theorem \ref{lem:linear-sketch} does not implement the `decoding' step which estimates $f( \overbar{z})$ from $A \cdot \overbar{z}$. This step depends on the problem being solved, however it is often very simple and thus implementable via a space-efficient SNN. E.g., in $\ell_p$ norm estimation one might  just have to compute the $\ell_p$ norm of $A \cdot  \overbar{z}$ \cite{indyk2006stable}. In frequency estimation, one might have to compute an average of a subset of entries in $A \cdot  \overbar{z}$ \cite{charikar2002finding}. 
 
Beyond our generic linear sketching reduction, we give neural solutions for two challenging problems in the insertion-only model, namely, distinct elements and median estimation. These simulation results are less general and provide several tools for bypassing critical obstacles that arise in streaming to SNN reductions. 

\paragraph{Distinct elements.} In the \emph{distinct elements problem} one must approximate the number of distinct items appearing in a data-stream with repeated items. It is well known that an exact solution by a single-pass streaming algorithm requires linear space. In fact, as we discuss later on, one can also show that the exact computation requires linear space in the SNN setting. Therefore, we restrict our attention to $(1+\epsilon)$ approximation for the number of distinct elements for any $\epsilon \in (0,1)$. This problem has been studied thoroughly in the streaming literature \cite{chassaing2007efficient,Bar-YossefJKST02,DurandF03, flajolet2007hyperloglog,KaneNW10,Blasiok18,IndykW03,Woodruff04,alon1999space}.

In this work, we provide an efficient neural implementation for the well-known LogLog streaming algorithm by \cite{DurandF03,flajolet2007hyperloglog}. While the LogLog and its improved variant the hyper-Loglog algorithm provide sub-optimal space bounds, they are commonly used in practice due to their simplicity. As we will see, they are efficiently implementable in the neural setting. In addition, we provide a nearly matching space lower bound. \vspace{-5pt}
\begin{theorem} [Neural Computation of Distinct Elements]\label{lem:disnct-elements}
For every $n \in \mathbb{N}$, $\epsilon, \delta \in (0,1)$, given $n$ input neurons representing the elements in $[n]$  
there exists a network $\cN$ with $\log n$ output neurons $\widetilde{O}(1/\epsilon^2)$ auxiliary neurons, and $O(\log \log n)$ persistence time that encode the logarithm of an $(1\pm \epsilon)$ approximation of the number of distinct elements in the current stream, with probability $1-\delta$. In addition, any SNN requires $\Omega(\log n+1/\epsilon^2)$ auxiliary neurons to compute an $(1\pm \epsilon)$ approximation for the problem, with constant probability. 
\end{theorem}
The lower bound is obtained via a communication complexity reduction that mimics the corresponding streaming reduction. We note that this reduction works perfectly, i.e., without any asymptomatic loss in the space-bound (compared to the streaming bound).

\paragraph{Count-Min sketch.} A common tool used in many of the streaming algorithms is the Count-Min sketch data structure, which maintains frequency estimates for all items in a stream. Count-Min sketch is in fact a linear sketch, and thus can be implemented via Theorem \ref{lem:linear-sketch}. However, it is not immediately clear how to implement certain important operations, like approximate frequency (count) queries via this reduction. We thus provide a direct implementation. Our implementation applies in the setting where there are $O(\log n)$ input neurons representing each insertion/deletion of an item $x \in [n]$. However, it can easily be extended to the setting in which there are $n$ input neurons, one for each item. 
\vspace{-5pt}
\begin{definition}[Count-Min Sketch~\cite{cormode2005improved}] Given parameters $\epsilon, \delta \in (0,1)$, the Count-Min sketch is a probabilistic data structure that serves as a frequency table of items in a stream. It supports two operations: (i) $\inc(x)$ increases the frequency of $x$ by one; (ii) $\countt(x)$ returns an $(1+\epsilon)$ approximation of the frequency of $x$ with probability $1- \delta$.
\end{definition}\vspace{-3pt}
For given parameters $\epsilon, \delta  > 0$, the Count-Min sketch data structure contains $\ell= O(\log 1/\delta)$ hash tables $T_1, \ldots T_\ell$ each with $b= O(1/\epsilon)$ bins, and each table $T_i$ is indexed using a different pairwise independent hash function $h_i$. 
The $\inc(x)$ operation applies $T_i[h_i(x)]\gets T_i[h_i(x)]+1$ for every $i \in [\ell]$. The $\countt(x)$ operation returns $\min_{i \in [\ell]} T_i[h_i(x)]$, which is shown to provide a good approximation for the frequency of $x$. 
The Count-Min data structure is used in many streaming algorithms including heavy-hitters, range queries, quantile estimation, and more. 
We provide an efficient neural implementation of a Count-Min sketch data structure, and show:
\begin{theorem}[Neural Implementation of Count-Min Sketch] \label{cor:min-sketch}
For every $n,m \in \mathbb{N}$ and $\epsilon, \delta \in (0,1)$ there exists a network $\cN$ with $\log n$ input neurons,
$O(1/\epsilon \cdot \poly(\log m, \log 1/\delta)))$ auxiliary neurons, and $\widetilde{O}(1)$ persistence time that implements a Count-Min sketch with approximation ration $(1+\epsilon)$ and success probability $1-\delta$, for an input stream of length at most $m$. 
\end{theorem}
Our neural implementation of the Count-Min sketch can immediately be used to give, e.g., a simple neural approximate heavy-hitters algorithm, which returns TRUE if a presented item has frequency $\ge m/k$ in a data-stream for some integer $k$, and FALSE if it has frequency $\le (1-\epsilon) m/k$. Setting $\epsilon' = O(\epsilon/k)$, a $\countt(x)$ query will return a frequency estimate $\ge m/k$ for any true heavy-hitter $x$ and $\le m/k$ for any $x$ with frequency $\le (1-\epsilon) m/k$. By keeping a counter for $m$ using $O(\log m)$ neurons and performing a comparison operation with the output of $\countt(x)$, we can thus solve the heavy hitters problem.
Other applications of Count-Min sketch require more complex processing of the data structure's output. To illustrate how this processing can be implemented efficiently in an SNN, we detail one such application, to median approximation.

\paragraph{Approximate median.} One of the most fundamental statistical measures of a data-stream is its quantile. The 1/2-quantile known as the median, attracts most attention in the streaming literature \cite{munro1980selection,manku1998approximate,charikar2002finding, ChenZ20}. Its non-linearity nature makes it considerably harder to maintain compared to its linear cousin, the mean. As in many other streaming problems, the exact computation of the median requires linear space both in the streaming and in the SNN setting (as will be discussed later on). This motivates the study of the relaxed $(1+\epsilon)$ approximation task. In the latter, the algorithm is allowed to output an item $j$ provided that the total number of items with value at most $j$ is in the range $[m/2-\epsilon m, m/2+\epsilon m]$.

Cormode and Muthukrishnan \cite{cormode2005improved} presented an elegant streaming algorithm for this problem using a space of $\widetilde{O}(1/\epsilon)$ bits. The algorithm is based on the Count-Min sketch data structure, combined with a \emph{dyadic decomposition} technique that is used in a number of other streaming algorithms. One of our key technical algorithmic contributions is in providing an efficient neural implementation of this algorithm. 
\begin{theorem}[Neural Computation of Approximate Median] \label{lem:approx-median}
For every $n, m \in \mathbb{N}$ and $\epsilon, \delta \in (0,1)$, there exists a neural network $\cN_{n,m}$ solving the $(1+\epsilon)$-approximate median problem using $O(1/\epsilon \cdot \poly(\log m, \log n, \log 1/ \delta))$ auxiliary neurons and persistence time $\widetilde{O}(1)$ with probability at least $1-\delta$. 
\end{theorem}
\vspace{-10pt}
\subsubsection{Streaming Lower Bounds Yield SNN Lower Bounds} Our second contribution focuses on Question \ref{ques:lower}, showing that space-efficient SNNs can be translated into space-efficient streaming algorithms, and thus that lower bounds in the streaming model imply lower bounds in the neural setting.
The underlying intuition for this transformation is based on the following observation.
\begin{observation}
A spiking neural network with deterministic edge weights, $n$ input neurons, and $S$ non-input neurons can be simulated by a streaming algorithm using $S$ bits of space. 
\end{observation}
In the SNN model, the spiking behavior of neurons in a given round depends only on the firing states of their incoming neighbors in the previous round. Thus, to simulate the behavior of the network as one pass over the data-stream, it is sufficient to maintain the firing states of all non-input neurons in the network, thus storing $S$ bits of information. When the edge weights of the network are randomly sampled such a small-space simulation becomes more involved. The explicit storage of all the edge weights might be too costly since there can be $\Omega(nS+S^2)$ edges in a network with $n$ input neurons and $S$ non-inputs. Nevertheless, we show that a small-space simulation is still possible using a pseudorandom number generator, if we pay an additive logarithmic overhead in the length of the stream and universe size.
\vspace{-1pt}
\begin{theorem}\label{lem:snn-rand-low-spaceIntro}
Any SNN $\mathcal{N}$ with $n$ input neurons, $S$ non-input neurons for $S=\poly(n)$, and $\poly(n)$ persistence time can be simulated over a data-stream of length $m$ using a total space of $O(S+\log(nm))$. The success guarantee of the simulation is $1-1/\poly(n,m)$. 
\end{theorem}
\vspace{-5pt}
Theorem \ref{lem:snn-rand-low-spaceIntro} is a powerful tool, since it lets us apply any streaming space lower bound (of which there are many) to give an SNN lower bound, with a loss of an $O(\log (nm))$ factor. In some cases, we can avoid this loss by more directly considering the lower bound technique. This is obtained when the streaming lower bounds are derived via a reduction to communication complexity with shared randomness that can be applied in the SNN setting with no loss. 
For example, using this tighter approach we show that our neural network for the distinct elements problem is nearly space-optimal (see Section~\ref{sec:distinct}).

\subsection{Preliminaries}

\paragraph{Spiking neural networks.}
A deterministic neuron $u$ is modeled by a \emph{deterministic} threshold gate. 
Letting $b(u)$ to be the threshold value of $u$, then $u$ outputs $1$ if the weighted sum of its incoming neighbors exceeds $b(u)$. A \emph{spiking neuron} is modeled by a \emph{probabilistic} threshold gate, which fires with a sigmoidal probability that depends on the difference between its weighted incoming sum and its threshold $b(u)$. 
\indent A \emph{Neural Network} (NN) $\cN =\langle X,Z,Y, w,b \rangle$ consists of $n$ input neurons $X=\{x_1, \ldots, x_{n}\}$, $m$ output neurons $Y=\{y_1, \ldots, y_{m}\}$, and $k$ auxiliary neurons $Z = \{z_1,...,z_{k} \}$. In spiking neural networks (SNN), the neurons can be either deterministic threshold gates or probabilistic threshold gates. The directed weighted synaptic connections between $V=X \cup Z \cup Y$ are described by the weight function $w: V \times V \rightarrow \mathbb{R}$. 
A weight $w(u,v) =0$ indicates that a connection is not present between neurons $u$ and $v$.  Finally, for any neuron $v$, the value $b(v) \in \mathbb{R}$ is the bias value (activation threshold). 
Additionally, each neuron is either inhibitory or excitatory:  if $v$ is inhibitory, then $w(v,u)\leq 0$ and if $v$ is 
excitatory, then $w(v,u)\geq 0$ for every $u$. This restriction arises from the biological structure of the neurons. 

\noindent\textbf{Network dynamics.}
The network evolves in discrete, synchronous rounds as a Markov chain. 
The firing status of every neuron $u$ in round $\tau$ denoted as $\sigma_\tau(u)$, depends on the firing status of its neighbors in round $\tau-1$. 
For each neuron $u$, and each round $\tau \ge 0$, let $\sigma_{\tau}(u)=1$ if $u$ fires (i.e., generates a spike) in round $\tau$. For every neuron $u$ and every round $\tau \ge 1$, let 
$\pot(u,\tau)= \sum_{v \in V}w(v,u)\cdot \sigma_{\tau-1}(v) -b(u)$ denote the membrane potential at round $\tau$. A deterministic threshold gate $u$ fires in round $\tau$ iff $\pot(u,\tau)\geq 0$. A probabilistic threshold gate fires in round $\tau$ with a probability that depends on  $\pot(u,\tau)$. All our network constructions in this work use deterministic threshold gates, and the randomness of the network comes from the randomized selection of the edge weights.

\noindent{\textbf{Neural networks for data-stream problems.}} A data-stream problem is defined by a relation $P_n \subset \mathbb{Z}^n \times \mathbb{Z}$. The length of the stream is upper bounded by some integer $m$. Each data-item is represented by a binary vector of length $n$. A value $i \in [1,n]$ is represented by having the $i^{th}$ input neuron fire while all other input neurons are idle. Each input is presented for some persistence time, at the end of which the output neurons of the network encode (in binary) the evaluation of a given relation over the current stream. To avoid cumbersome notation, we may assume that $m$ and $n$ are powers of $2$.

\subsection{Basic Tools} \label{sec:tools}
Our constructions are based on several neural network modules. We start by describing existing tools we will be using, and then describe additional new tools. 

\paragraph{Neural timers and counters.}
For a given time parameter $t$, a neural timer $\mathcal{NT}_t$ is an SNN network that consists of an input neuron $x$, an output neuron $y$, and additional auxiliary neurons. The network satisfies that in every round $\tau$, $y$ fires in round $\tau$ iff there exists $\tau' \in [\tau - t,\tau]$ such that $x$ fired in round $\tau'$.  
It is fairly easy to design a neural timer network with $O(t)$ auxiliary neurons. \cite{HitronP19} presented a construction of a more succinct network $\mathcal{NT}_t$ with only $O(\log t)$ neurons.
In the related setting of neural counting, the network is required to encode the \emph{number} of firing events of its input neuron within a given time window.  Specifically, given a time parameter $t$, a \emph{neural counter network} $\mathcal{NC}_t$ has a single input neuron $x$, and $\lceil \log t \rceil$ output neurons that encode the number of firing events of $x$ within a span of $t$ rounds. 
\begin{fact}\label{fc:neural-counter}\cite{lynch2019integrating,HitronP19}
For every integer parameter $t$, there exist (i) a neural timer network $\mathcal{NT}_t$ with $O(\log t)$ neurons, and (ii) a neural counter network $\mathcal{NC}_t$ with $O(\log t)$ auxiliary neurons, such that for every round $i$, the output neurons encode $f_i$ by round $i+O(\log t)$ where $f_i$ is the number of firing events up to round $i$. Both networks $\mathcal{NT}_t$ and $\mathcal{NC}_t$ are deterministic.
\end{fact}
%
\paragraph{Maximum computation.}
Given a neural representation of $m$ elements $x_1, \ldots, x_m$ in $[n]$, it is required to design a neural network that computes the maximum value $x^*=\max_i x_i$. The network has an input layer of $m\cdot \log n$ neurons that represent the elements $x_1, \ldots, x_m$, and an output layer of $\log n$ neurons that should encode the value of the maximum value $x^*$.
Maass~\cite{maass2000computational} presented a network construction with $O(m^2+m \cdot \log n)$ auxiliary neurons. 
In the high level, in this network for every pair of elements $x_i,x_j$, there is a designated comparison neuron $c_{i,j}$ which fires if and only if $x_i \geq x_j$. The output is then computed using additional $m$ neurons $g_1, \ldots g_m$ where $g_i$ fires if and only if all the comparison neurons $c_{i,1}, \ldots, c_{i,m}$ fired. We have:
\begin{fact}\label{fc:max}\cite{maass2000computational}
Given $m$ vectors of neurons $\overbar{x}_1, \ldots, \overbar{x}_m$ each of size $\log n$, there exists a deterministic neural network with $\log n$ output neurons $\overbar{y}$ and $O(m^2+m \log n)$ auxiliary neurons, such that if the input neurons encode the values $x_1, \ldots , x_m$ in round $t$, the output neurons $\overbar{y}$ represents the value $\max_{i}x_i$ by round $t+O(1)$.
\end{fact}
Upon very small modifications, the same network solution can be adapted to compute the minimum and the median elements.
We next describe \emph{new} tools introduced in this work which will be heavily used in our constructions.

\paragraph{Potential encoding.}
Our SNN constructions are based on a module that encodes the potential $p$ of a given neuron $x$ by its binary representation using $O(\log p)$ neurons. We will use this modules in the constructions of Theorem~\ref{lem:linear-sketch} and Lemma~\ref{lem:hash}. 
\begin{lemma} \label{clm:binary-potential}
Let $x$ be a deterministic neuron such that $\pot(x,t')\leq 2^{\ell}$ for every $t' \in [t, t+O(\ell)]$ for some integer $\ell \in \mathbb{N}_{>0}$. 
There exists a deterministic network $\POT_{\ell}(x)$ which uses $\ell$ identical copies of $x$ (with the same input and bias), $2 \ell$ auxiliary neurons, and $\ell$ output neurons $y_0 \ldots y_{\ell-1}$ that encodes $\pot(x,t)$ in a binary form within $O(\ell)$ rounds. 
\end{lemma} 
\begin{proof}[Proof of Lemma \ref{clm:binary-potential}]
We begin with describing the network. 
\begin{itemize}
\item The input to the network are  $\ell$ excitatory copies of $x$ denoted as $x_0, \ldots, x_{\ell-1}$. Each $x_i$ has all the incoming edges and bias as the neuron $x$. 
\item There are $\ell$ inhibitory neurons $r_0, \ldots r_{\ell-1}$ each with bias $-1$ and no incoming edges. Hence, these inhibitory neurons keep on firing on every round. Every neuron $r_i$ is connected to $x_i$ with weight $w(r_i,x_i)=-(2^i-1/2)$.
\item There are additional $\ell-1$ inhibitory neurons $v_1, \ldots, v_{\ell -1}$. Each $v_i$ has an incoming edge from $x_i$ with weight $w(x_i,v_i)=1$ and bias $b(v_i)=1$. Additionally, every $v_i$ has $i-1$ outgoing edges to $x_0, \ldots x_{i-1}$ with weight $w(v_i,x_j)=-2^i$. 
\item Let $y_0, \ldots, y_{\ell-1}$ be the output neurons of the network. Each $y_i$ has an incoming edge from $x_i$ with weight $w(x_i,y_i)=1$ and bias $b(y_i)=1$. Hence, $y_i$ fires in some round $t$ iff $x_i$ fired in the previous round.
\end{itemize}
See Figure~\ref{fig:potential} for an illustration of the network.

\begin{figure}[h!]
\begin{center}
\includegraphics[scale=0.4]{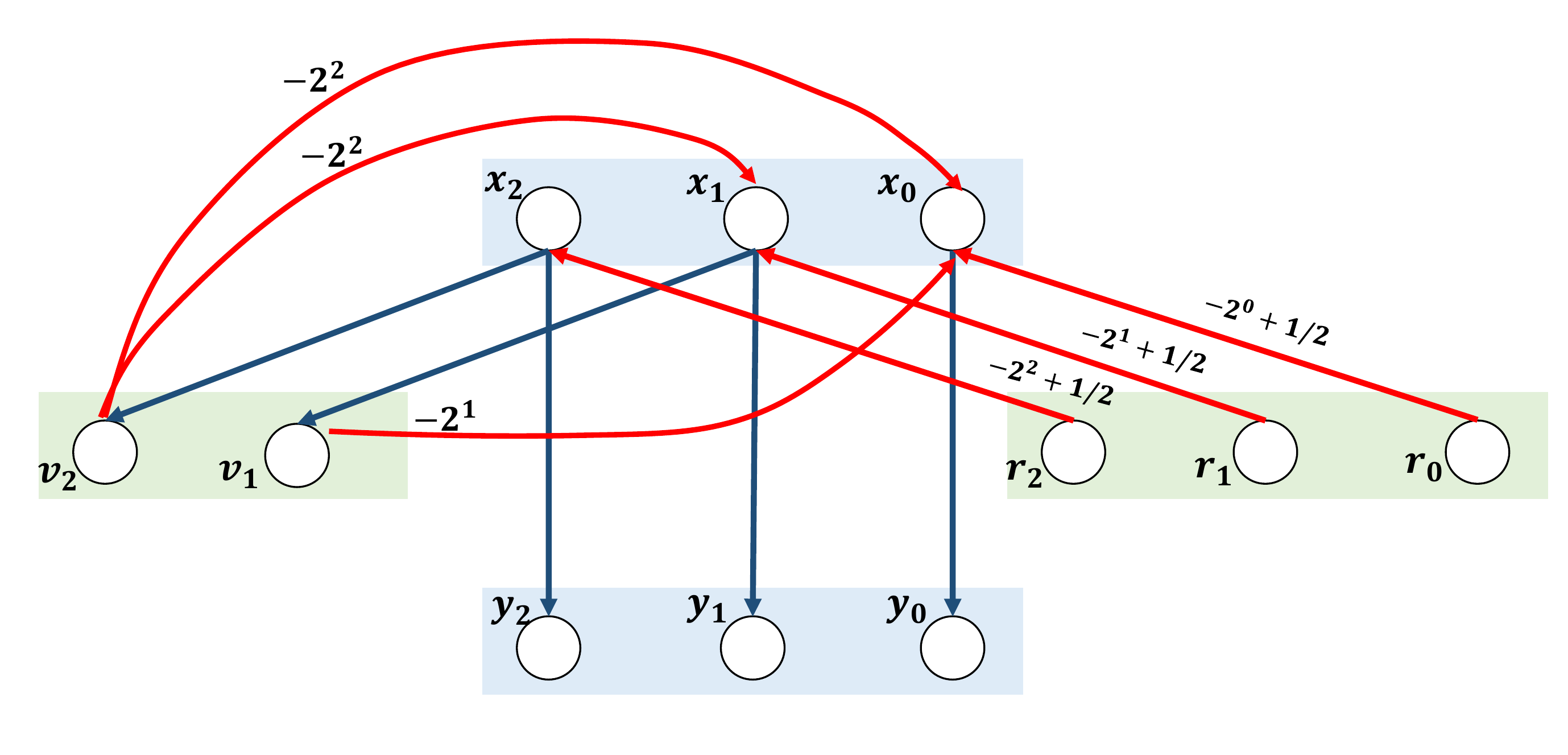}
\caption{ An illustration of the potential encoding module. \label{fig:potential}}
\end{center}
\end{figure}

\paragraph{Correctness analysis:}
Let $\lfloor \pot(x,t) \rfloor = \sum_{i=0}^{\ell-1}a_i\cdot 2^i$ be the potential of $x$ in some round $t_0$. We assume the input is persistent for at least $2\ell$ rounds. This implies that the potential of $x$ does not change for at least $2\ell$ rounds.
We will show by induction on $i$ that starting round $t_0+2\cdot i$ the output neuron $y_{\ell-i-1}$ fires iff $a_{\ell-i-1} = 1$. Base case: for neuron $x_{\ell-1}$, the only inhibitor inhibiting $x_{\ell-1}$ is $r_{\ell-1}$ with weight $-(2^{\ell-1}-1/2)$. Hence, if $a_{\ell-1}=1$ then starting round $t_0+1$ the potential of $x_{\ell-1}$ is at least $\sum_{i=0}^{\ell-1}a_i\cdot 2^i-2^{\ell-1}+1/2 \geq 1/2>0$. Thus, $x_{\ell-1}$ fires starting round $t_0+1$ and therefore $y_{\ell-1}$ fires starting round $t_0+2$. On the other hand, if $a_{\ell-1}=0$ then starting round $t_0+1$ the potential of $x_{\ell-1}$ is given by $\sum_{i=0}^{\ell-2}a_i\cdot 2^i-2^{\ell-1}+1/2 \leq (2^{\ell-1}-1)-2^{\ell-1}+1/2<0$ and therefore starting round $t_0+1$ the neuron $x_{\ell-1}$ is idle and $y_{\ell-1}$ does not fire starting round $t_0+2$. Assume the claim is correct for neurons $y_{\ell-1}, \ldots, y_{\ell-i}$ and consider neuron $y_{\ell-i-1}$. 

By the induction assumption by round $t_0+2\cdot i$ the neurons $y_{\ell-1}, \ldots y_{\ell-i}$ encode $a_{\ell-1}, \ldots a_{\ell-i}$. By the definition of the network, we conclude that for the inhibitors $v_{\ell-1}, \ldots v_{\ell-i}$, each $v_j$ fires starting round $t_0+2i$ iff $a_j=1$. Hence, the potential of $x_{i-1}$ in round $t_0+2i+1$ is equal to $\sum_{i=0}^{\ell-i-1}a_j\cdot 2^j - 2^{\ell-i-1}+1/2$. Therefore $x_{\ell-i-1}$ fires starting round $t_0+2i+1$ iff $a_{\ell-i-1}=1$ and $y_{\ell-i-1}$ encodes $a_{i-1}$ starting round $t_0+2(i+1)$.
\end{proof}

\paragraph{Implementing pairwise independent hash functions.} 
Many streaming algorithms in the insertion only model are based on the notion of pairwise independent hash functions. \vspace{-5pt}
\begin{definition}[Pairwise Independence Hash Functions]
A family of functions $\mathcal{H}:[a] \to [b]$ is \emph{pairwise independent} if for every $x_1 \neq x_2 \in [a]$ and $y_1, y_2 \in [b]$, we have:
$$\Pr[h(x_1) = y_1 \mbox{~and~} h(x_2) = y_2] = 1/b^2.$$
\end{definition}
\vspace{-5pt}
For ease of notation, assume that $a,b$ are powers of $2$.
\begin{definition}[Pairwise Independence Hash SNN]
Given two integers $a,b$, a \emph{pairwise independent hash network} $\mathcal{N}_{a,b}$ is an SNN with an input layer of $\log a$ neurons, an output layer of $\log b$ neurons, and a set of $s$ auxiliary spiking neurons. For every input value $x$ presented at round $t$, let $\mathcal{N}(x)$ be the value of the output layer after $\tau_{a,b}$ rounds. Then, for every $x \neq x' \in [a]$, it holds that $\Pr[\mathcal{N}(x)=\mathcal{N}(x')]=1/b$. 
\end{definition}\vspace{-5pt}
We show a neural network implementation of a pairwise independent hash function using the construction of pairwise hash function by \cite{carter1979universal}.
\vspace{-5pt}
\begin{lemma} [Neural Implementation of Pairwise Indep. Hash Function]\label{lem:hash}
For every integers $a,b$, there exists a pairwise independent hash network $\cN_{a,b}$ with $s=O(\log b \cdot \log \log a)$ auxiliary neurons and $O(\log \log a)$ persistence time. 
\end{lemma}
\begin{proof}[Proof of Lemma \ref{lem:hash}] 
Our neural network implements the well-known construction of a pairwise independent hash function of \cite{carter1979universal}.
In this construction, the input is a binary vector of length $c=\log a$, and the output is a binary vector of length $d = \log b$. Letting $H$ be a binary $c \times d$ matrix sampled uniformly at random, define $h(x) = H\cdot x$, 
where the addition operations are defined over the field $\mathbb{F}_2$. 
\cite{carter1979universal} showed that for every $x \neq y$, $\Pr[h(x)=h(y)]=1/2^d=1/b$.

For each $i \in [\log b]$, the network connects all input neurons to one intermediate neuron $r_i$ with random weights in $\{0,1\}$ and bias $b(r_i)=0$. Hence, the potential of $r_i$ is equal to the multiplicity of $\overbar{x}$ with a binary random vector. Next, the network extracts the potential of $r_i$ using the sub-network $\POT(r_i)$ defined in Lemma~\ref{clm:binary-potential}. For that purpose, it introduces $O(\log \log a)$ copies of the neuron $r_i$. Next, in order to compute the addition in $\mathbb{F}_2$, the network computes the parity of the potential of $r_i$ by connecting the least significant bit of the output of the sub-network $\POT(r_i)$ to the $i^{th}$ output neuron $y_i$. 
The correctness of the constructed network follows from Lemma~\ref{clm:binary-potential}.
\end{proof}
\section{Linear Sketching}\label{sec:Lin-Sketch}
A linear sketching algorithm is a streaming algorithm in which the state of the algorithm at time $t$ is a linear function of the updates seen up to time $t$. We start with a formal definition.

\begin{definition}[Linear Sketching Algorithm, \cite{kapralov2020fast}]
A linear sketching algorithm $\mathcal{L}$ gives a method for processing a vector $\overbar{x} \in \mathbb{R}^n$.
The algorithm is characterized by a (typically randomized) sketch
matrix $A \in \mathbb{R}^{r \times n}$, and by a possibly randomized decoding function $f : \mathbb{R}^r \to O$
where $O$ is some output domain. Algorithm $\mathcal{L}$ is executed by first computing $A \cdot \overbar{x}$ and then outputting 
$f(A \cdot \overbar{x}$). Note that $f$ only takes $A \cdot \overbar{x}$ as input, $f$ cannot depend on $A$ in any other way, e.g. it cannot share randomness with $A$.
\end{definition}
Linear sketching algorithms provide the state-of-the-art space bounds for a large collection of problems in the turnstile model. 

\vspace{-10pt}
\paragraph{The challenge and our approach.} 
Throughout we assume the sketching matrix is integral, i.e., $A \in \mathbb{Z}^{r \times n}$, which captures most of the classic implementations in the turnstile model.
We start by describing a straw man approach for computing the value $A \overbar{x}$ in the neural setting: Take a single-layer neural network with an input layer of length $n+1$ and an output layer of length $r$. Specifically, the input layer contains $n$ neurons $x_1,\ldots,x_n$ that represent the absolute value of the update, and an additional \emph{sign} neuron that indicates the sign of the update. For example, an update vector $[0,0,-1,0]$ is represented by letting $x_3=1$, $s=1$ and $x_1,x_2,x_4=0$. The output layer is defined by $r$ output neurons $y_1,\ldots, y_r$. The edge weights are specified by the matrix $A$ where $w(x_j,y_i)=A_{i,j}$. It is then easy to verify that the weighted sum of the incoming neighbors of each neuron $y_j$ (i.e., its potential) is the value of the $j^{th}$ entry in the vector $A\overbar{x}$. 

This naive description fails for various reasons. First, from a biological perspective, each input neuron can be either inhibitory or excitatory. This implies that the sign of the outgoing edge weights of a given neuron must be either a plus (excitatory) or a minus (inhibitory). Mathematically, this requires the sketch matrix $A$ to be sign-consistent (i.e., the sign of all entries in a given raw are either a plus or a minus). 
However, in general, the given sketch matrix might not be sign-consistent.
The second technicality is that the neurons $y_1,\ldots,y_n$ have a \emph{binary} output (either firing or not) rather then an \emph{integer} value. The third aspect to take into account is concerned with the update mechanism. Specifically, given a stream of data items, one should make sure that each data item would be processed exactly \emph{once} by the network. This requires a more delicate update mechanism. 

In the high-level, we handle the sign-consistency challenge by dividing the sketch matrix $A$ into a non-negative matrix $A^{+}$ and a non-positive matrix $A^{-}$ where $A = A^{+}-A^{-}$. Then, given a new update $(\overbar{x}, s)$, the network computes $A\overbar{x}$ and $-A\overbar{x}$ using $A^{+}\overbar{x}$ and $A^{-}\overbar{x}$. The final output $A\overbar{x}$ is computed by using these values combined with the sign neuron $s$. 
To handle the second challenge,  we use the module of Lemma~\ref{clm:binary-potential} to translate the \emph{potential} of each output neuron $y_j$ (corresponding to the $j$'th bit in the sketch) into its binary representation. The output layer consists of $O(r\log n)$ output neurons that encode the value of the current $r$-length sketch.

\paragraph{Network Description.} 
Let $(\overbar{x}=(x_1, \ldots ,x_n),s)$, be the input neurons of the linear sketch network $\cN$, where $\overbar{x}$ represents the current update and $s$ represents the update sign, firing if the update is negative. Let $A \in \mathbb{Z}^{r \times n}$ be the sketch matrix of the sketching algorithm to be implemented. We denote the multiplication of the input vector $x'$ represented by $(\overbar{x},s)$ and $A$ by $A \circ (\overbar{x},s)$. This notation is needed since $\overbar{x}$ represents only the absolute value of the update. 
The output layer of the network consists of $r$ vectors $\overbar{y}_1, \ldots , \overbar{y}_r$ and $r$ sign-neurons 
$s_1, \ldots s_r$. Each vector $\overbar{y}_j$ contains $O(\log n)$ neurons that are used to encode in binary the absolute value of the $j^{th}$ entry in the output sketch. The sign of this entry is specified by the sign neuron $s_j$.  

In order for the output neurons to continue presenting the correct value throughout the execution, all output neurons $\overbar{y_1}, \ldots , \overbar{y_r}, s_1, \ldots s_r$ have self-loops with large positive weights. For each output vector $\overbar{y}_i$ with sign neuron $s_i$ the algorithm introduces an equivalent vector of inhibitory neurons $\overbar{y}'_i$ where each neuron $y'_{i,j}$ serves as an $\AND$ gate between $y_{i,j}$ and $s_i$. In addition, for each of the vectors $\overbar{y}_{i}$, $\overbar{y}'_{i}$, an inhibitor and excitatory copies are introduced, in which each neuron has the same incoming edges and biases as its corresponding neuron. These copies will assist us in case the new value after the current update will be negative. In our network description, all neurons, unless specified otherwise, are excitatory by default.
\vspace{-10pt}
\paragraph{(1) Matrix Multiplication.} Let $A^+$ (resp., $A^{-}$) be the matrix containing the non-negative (resp., non-positive) entries of $A$, where 

$$(A^{+})_{i,j} =
\begin{cases}
A_{i,j},& \text{if } A_{i,j} \geq 0 \\
0,& \text{Otherwise}~.
\end{cases} 
\mbox{~~~and~~~}
	(A^{-})_{i,j} =
	\begin{cases}
	-A_{i,j},& \text{if } A_{i,j} < 0 \\
	0,& \text{Otherwise}~.
	\end{cases}
	$$
Hence, both $A^{+}$ and $A^{-}$ are non-negative matrices and $A\overbar{x}=A^{+}\overbar{x}-A^{-}\overbar{x}$.
The network contains two vectors of neurons $a^{+}$ and $a^-$ of length $r$ that are connected to the input neurons with weights $w(x_i,a^+_{j})=A^{+}_{i,j}$,  $w(x_i,a^{-}_{j})=A^{-}_{i,j}$ for every $i \in \{1,\ldots, n\}$ and $j \in \{1,\ldots,r\}$.  For each neuron $a^+_{j}$ and each neuron $a^{-}_{j}$, there are $\lceil \log \ell \rceil$ copies, in order to describe their potential by a binary vector denoted as $z^+_{j}$, $z^{-}_{j}$. This can be done using the potential-encoding module of Lemma~\ref{clm:binary-potential}. In addition, each vector $z^+_{j}$, $z^{-}_{j}$ has an inhibitory vector copy $z'^+_{j}$, $z'^{-}_{j}$.
\vspace{-10pt}
\paragraph{(2) Computing $A\overbar{x}$ and $-A\overbar{x}$.}
In order to calculate $A\overbar{x}$ (and $(-A\overbar{x}$)), for each $i \in [r]$ the algorithm introduces two additional neurons $a_{p,i}, a_{n,i}$ such that the potential of $a_{p,i}$ equals $A\overbar{x}$, and the potential of  $a_{n,i}$ equals $(-A\overbar{x})$. Neuron $a_{p,i}$ has incoming edges from $z^+_i$ and the inhibitor vector $z'^{-}_{i}$ with weights $w(z^{+}_{i})_j,a_{p,i})= 2^{j-1}$ and $w((z'^{-}_{i})_j),a_{p,i})= -2^{j-1}$. Similarly, neuron $a_{n,i}$ has incoming edges from $z^{-}_{i}$ and the inhibitors $z'^{+}_{i}$ with weights $w((z^{-}_{i})_j,a_{n,i})= 2^{j-1}$ and $w((z'^{+}_{i})_j,a_{n,i})= -2^{j-1}$. 
	We then introduce $\log \ell$ copies for each $a_{p,i}$ and $a_{n,i}$, and extract their potential into binary vectors using Lemma~\ref{clm:binary-potential}. The output neurons of the sub-networks $\POT(a_{p,i})$ and $\POT(a_{n,i})$ together with the sign neuron $s$, are connected to four neural vectors $\overbar{d}_{p,i}$, $\overbar{d}'_{p,i}$, $\overbar{d}_{n,i}$, $\overbar{d}'_{n,i}$ as follows. 
	
The vector $\POT(a_{p,i})$ is connected to $\overbar{d}_{p,i}$ and $\overbar{d}'_{p,i}$, where for every $j \in [\log \ell]$, the $j^{th}$ neuron in $\overbar{d}_{p,i}$ fires only if the $j^{th}$ neuron in $\POT(a_{p,i})$ fire.
In addition, the $j^{th}$ neuron of $\overbar{d}'_{p,i}$ is an $\AND$ gate of 
$s$ and the $j^{th}$ neuron of $\POT(a_{p,i})$. In a similar manner, the vector $\POT(a_{n,i})$ is connected to the inhibitory neurons $\overbar{d}'_{n,i}$ where for every $j \in [\log \ell]$, the $j^{th}$ neuron in $\overbar{d}'_{n,i}$ fires only if the $j^{th}$ neuron in $\POT(a_{p,i})$ fire. The $j^{th}$ neuron of the excitatory neurons $\overbar{d}_{n,i}$ is an $\AND$ gate of $s$ and the $j^{th}$ neuron of $\POT(a_{n,i})$.
	
Note that because $(A\overbar{x})_i>0$ or maybe $(-A\overbar{x})_i>0$ but not both, for every coordinate $i$ either $\overbar{d}_{p,i}$, $\overbar{d}'_{p,i}$ contain firing neurons or $\overbar{d}_{n,i}$, $\overbar{d}'_{n,i}$ contain firing neurons but not both.
\vspace{-10pt}
\paragraph{(3) Computing $A \circ (\overbar{x},s)$.}
For every coordinate $i$, there is an intermediate neuron $q_i$ whose potential corresponds to the value of the $i^{th}$ coordinate in the updated vector. The neuron $q_i$ has positive incoming edges from the neurons in $\overbar{y}_i$, $\overbar{y}'_i$ with weights $w(y_{i,j},q_i)=2^{j-1}$ and $w(y'_{i,j},q_i)=-2 \cdot 2^{j-1}$ respectively. Hence, in case the sign neuron $s_i$ is idle, the output neurons contribute $\dec(\overbar{y}_i)$ to the potential of $q_i$. 
In case the sign neuron $s_i$ fires, the output neurons contribute $-\dec(\overbar{y}_i)$ to the potential of $q_i$.
	
The next step is to add the value $A\overbar{x}$ to the potential value of every $q_i$. To do that, the network connects the vectors $\overbar{d}_{p,i}$, $\overbar{d}'_{p,i}$, $\overbar{d}_{n,i}$, $\overbar{d}'_{p,i}$ to $q_i$ in the following manner. 
The vectors $\overbar{d}_{p,i}$, $\overbar{d}'_{p,i}$ are connected to $q_i$ with weights $w((\overbar{d}_{p,i})_j,q_i) =2^{j-1}$, and $w((\overbar{d}'_{p,i})_j,q_i) = -2 \cdot 2^{j-1}$. 
Hence, in case where the sign neuron $s$ is idle ,these neurons contribute $\dec(\overbar{d}_{p,i})$ to the potential of $q_i$. In case where the sign neuron $s$ fires, these neurons contribute $-\dec(\overbar{d}_{p,i})$. Recall that the neural vectors $\overbar{d}_{p,i}$, $\overbar{d}'_{p,i}$ have firing neurons only if $(A\cdot \overbar{x})_i > 0$, and in this case $\dec(\overbar{d}_{p,i})= (A\cdot \overbar{x})_i$. 
	
In the same manner, the vectors $\overbar{d}_{n,i}$ and $\overbar{d}'_{n,i}$ are connected to $q_i$ with weights $w((\overbar{d}_{n,i})_j,q_i) =2 \cdot 2^{j-1}$, and $w((\overbar{d}'_{n,i})_j,q_i) = - 2^{j-1}$.
Recall that the vectors $\overbar{d}_{n,i}$, $\overbar{d}'_{n,i}$ have firing neurons only if $(A\cdot \overbar{x})_i < 0$ and in this case $\dec(\overbar{d}_{p,i})= -(A\cdot \overbar{x})_i$.
	
Thus, by the above description the potential of $q_i$ encodes the $i^{th}$ coordinate of the updated output vector. 
To support the case where the potential of $q_i$ is negative, there is an equivalent neuron $q'_i$ whose potential 
is the additive inverse of the potential of $q_i$. This can be obtained by using 
the relevant excitatory and inhibitory copies of the vectors $\overbar{y}_i$, $\overbar{y'}_i$, $\overbar{d}_{p,i}$, $\overbar{d}'_{p,i}$, $\overbar{d}_{n,i}$, $\overbar{d}'_{p,i}$ that are connected to $q'_i$ with the corresponding weights (as used in $q_i$ up to a flip in the sign).
\vspace{-10pt}
\paragraph{(4) Updating the Output Sketch (Exactly Once).} The network updates the output vectors $\overbar{y}_1, \ldots , \overbar{y}_r$, $s_1, \ldots s_r$ in the following manner.
Using Lemma~\ref{clm:binary-potential} and $\log \ell$ identical copies of $q_i$ and $q'_i$, it extracts their potential into vectors of neurons denoted as $Q_i$, $Q'_i$. The algorithm connects each neuron $q'_i$ to the sign neuron $s_i$ and the vectors $Q_i$, $Q'_i$ to the output vector $\overbar{y}_i$ via a delay chains of size $O(1)$, where $y_{i,j}$ has an incoming large positive weight from the $j^{th}$ neurons of $Q_i$ and $Q'_i$. The reason we use a delay chain is that we wish to add the update $A\overbar{x}$ to the output $\overbar{y}$ after the previous value is deleted, and only \emph{once}. For that purpose, the following \emph{reset mechanism} is added. 
	
The algorithm connects the input neurons $\overbar{x}$ to an intermediate excitatory neuron $r_0$ which serves as a simple $\OR$ gate. Let $\tau$ be an upper bound on the number of rounds between the first round the input $\overbar{x}$ is presented and the round in which the neurons in $Q_1 , \ldots Q_{r}$ and $Q'_1 \ldots Q'_{r}$ output the desired outcome as specified in Lemma~\ref{clm:binary-potential}. The neuron $r_0$ is connected to a chain $C$ of size $\tau+1$ where each neuron $c_j$ in $C$ has an incoming edge from $c_{j-1}$ with weight $1$ and bias $1$. The neuron $c_{\tau}$ is then connected to the first neurons in the delay chains connected to $q_i, Q_i, Q'_i$ for every coordinate $i$, where these neurons serve as $\AND$ gates of the input from $c_{\tau}$ and the corresponding neuron in $q_i, Q_i, Q'_i$. The last neuron in $C$, (i.e. $c_{\tau+1}$) is an inhibitory neuron with outgoing edges to all neurons in the network besides the delay chains with large negative weights (including the output neurons). 
Hence, once the network is reset the algorithm will update the output neurons \emph{once}, and all neurons will be idle until the next update. Figure~\ref{fig:sketch} illustrates the constructed network. 

\begin{figure}[h!]
\begin{center}
\includegraphics[scale=0.5]{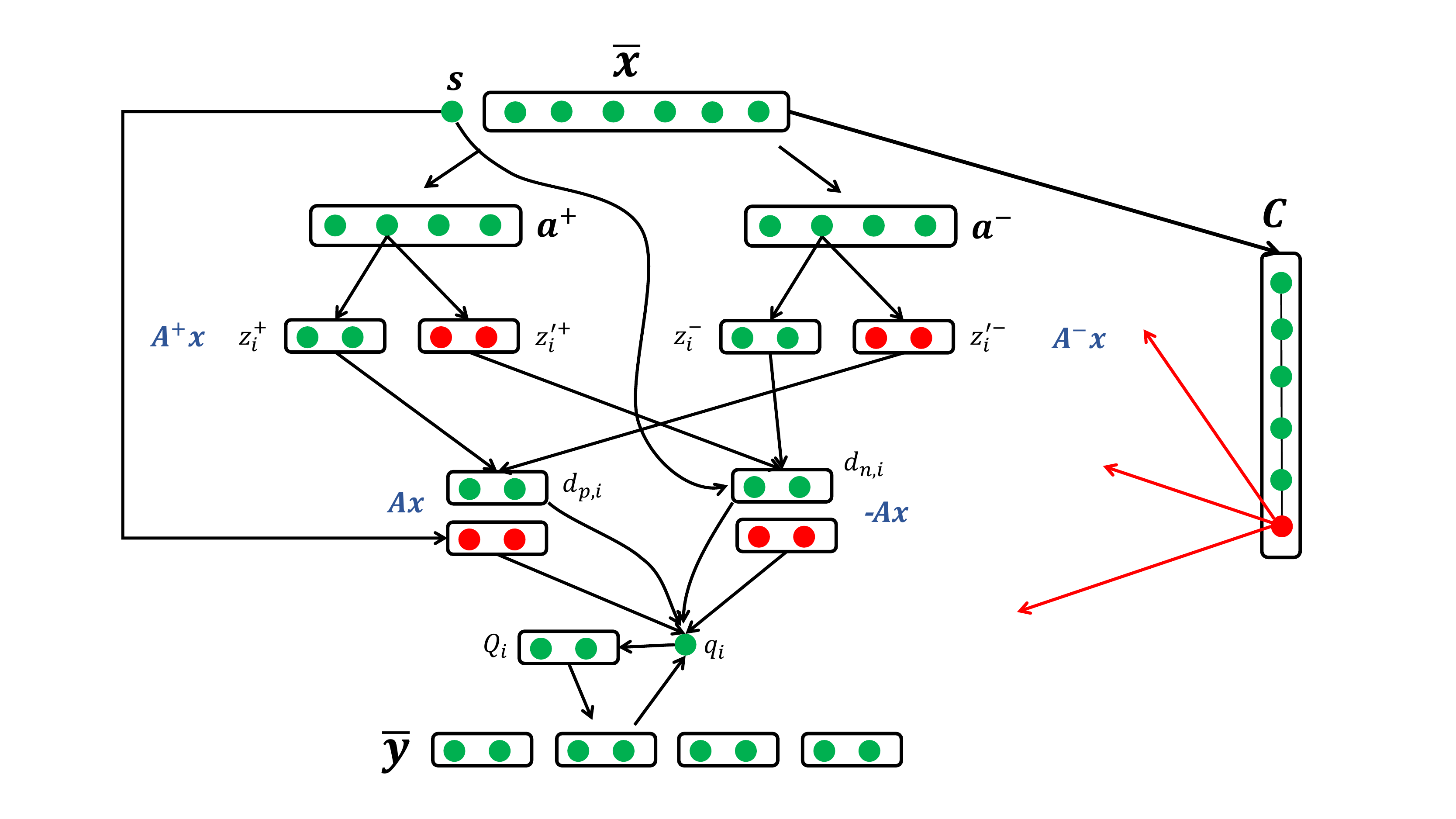}
\caption{\label{fig:sketch} \small{An illustration of the linear sketching network.
The red circles represent inhibitor neurons, and the green circles represent excitatory neurons. In the first layer, the potential values of the vector $a^+$ ($a^-$) represent $A^+\overbar{x}$ ($A^-\overbar{x}$). For simplicity from that point, we focus on updating the $i^{th}$ coordinate of the output neurons. In the third layer the potential value of $(A^+\overbar{x})_i$ ($(A^-\overbar{x})_i$) are extracted into a binary vector $z^+_i$ ($z^-_i$). These neurons are then used to compute $A\overbar{x}$ and $-A\overbar{x}$ which are represented in $d_{p,i}$ and $d_{n,i}$ and their inhibitory copies. The inhibitory neurons $d'_{p,i}$ and the excitatory neurons $d_{n,i}$ fire only if $s=1$. In the next layer, these neurons and the output neurons $y_i$ are connected to $q_i$ such that the potential of $q_i$ is equal to the updated value. The potential of $q_i$ is extracted into a binary vector $Q_i$ that is connected to the output neurons. The chain $C$ is responsible to schedule the update so that it will occur only once, and only after the reset of $y$.}}
\end{center}
\end{figure}

\paragraph{Correctness.}
Let $\overbar{x},s$ be an update presented in round $\tau_0$.
\begin{observation} \label{obs:A+-values}
For every coordinate $i \in [r]$, the firing neurons in $z^+_i$ encode the binary representation of $(A^{+}\overbar{x})_i$, and the firing neurons in $z^-_i$ encode $(A^{-}\overbar{x})_i$  by round $\tau_1 = \tau_0+O(\log \ell)$. 
\end{observation} 
\begin{proof}
Starting round $\tau_0+1$ due to the edge weights between $\overbar{x}$ and $a^+$ ($a^-$), the potential value of $(a^+)_i$ is equal to $(A^{+}\overbar{x})_i$ and the potential value of $(a^-)_i$ is equal to $(A^{-}\overbar{x})_i$. Since all entrees in $A^+$ and $A^-$ are non-negative, these potential values are non-negative. Hence, by Lemma~\ref{clm:binary-potential} the neurons $z^+_i$ and $z^-_i$ holds a binary representation of these potential within $O(\log \ell)$ rounds. 
\end{proof}
We now show that the potential of $q_i$ is equal to the updated value within $O(\log \ell)$ rounds. Let $\widehat{y}$ be the values encoded in the output neurons in round $\tau_0$.
\begin{claim} \label{clm:sketch-potential}
For every coordinate $i \in [r]$, the potential of $q_i$ is equal to the updated value $(\widehat{y} + (1-2s)A\overbar{x})_i$ by round $\tau_2 = \tau_0+ O(\log \ell)$. 
\end{claim}
\begin{proof}
By Observation~\ref{obs:A+-values}, the neurons $z^+_i$, $z^-_i$  encodes $(A^{+}\overbar{x})_i$ and $(A^{-}\overbar{x})_i$ respectively by round $\tau_1 =\tau_0+O(\log \ell)$. 
Since $A = A^+-A^-$, the potential value of $a_{p,i}$ equals $(A\overbar{x})_i$ and the potential value of $a_{n,i}$ is equal to $-A\overbar{x})_i$ by round $\tau_1+1$. Hence, by Lemma~\ref{clm:binary-potential} for some round $\tau'= \tau_1 + O(\log \ell)$, if $(A\overbar{x})_i \geq 0$, the output neurons of $\POT(a_{p,i})$ encodes $(A\overbar{x})_i$ by round $\tau'$ and if $(A\overbar{x})_i < 0$ the output neurons of $\POT(a_{p,i})$ are idle. Similarly, if $(A\overbar{x})_i \leq 0$, the output neurons of $\POT(a_{n,i})$ encodes $(-(A\overbar{x})_i)$ by round $\tau'$ and if $(A\overbar{x})_i > 0$ the output neurons of $\POT(a_{n,i})$ are idle. 

Next, we calculate the contribution of the vectors $\overbar{d}_{n,i}, \overbar{d}'_{n,i}$ , $\overbar{d}_{p,i}, \overbar{d}'_{p,i}$ to the potential of $q_i$. 
Starting at round $\tau'+1$, if $(A\overbar{x})_i \geq 0$ then $\dec(\overbar{d}_{p,i}) = \dec(\POT(a_{p,i}))= (A\overbar{x})_i$ and $\dec(\overbar{d}_{n,i})=\dec(\overbar{d}'_{n,i})=0$. In case $s=0$, then $\dec(\overbar{d}'_{p,i}) = 0$ and therefore the neurons $\overbar{d}_{n,i}, \overbar{d}'_{n,i}$ , $\overbar{d}_{p,i}, \overbar{d}'_{p,i}$ contribute $(A\overbar{x})_i$ to the potential of $q_i$. In case $(A\overbar{x})_i \geq 0$ and $s=1$ then $\dec(\overbar{d}'_{p,i}) = (A\overbar{x})_i$ and by the definition of the weights from $\overbar{d}'_{p,i}$ to $q_i$, the neurons $\overbar{d}'_{p,i}$ contribute $-2(A\overbar{x})_i$ to the potential of $q_i$. Hence, if $(A\overbar{x})_i \geq 0$ and $s=1$ the neurons $\overbar{d}_{n,i}, \overbar{d}'_{n,i}$ , $\overbar{d}_{p,i}, \overbar{d}'_{p,i}$ contribute $(A\overbar{x})_i-2(A\overbar{x})_i=-(A\overbar{x})_i$ to the potential.

On the other hand, in case $(A\overbar{x})_i \leq 0$ then $\dec(\overbar{d}_{p,i})=\dec(\overbar{d}'_{p,i})=0$ and $\dec(\overbar{d}'_{n,i})=\dec(\POT(a_{n,i}))=-(A\overbar{x})_i$. If $s=0$ then $\dec(\overbar{d}_{n,i})=0$ and the neurons $\overbar{d}_{n,i}, \overbar{d}'_{n,i}$ , $\overbar{d}_{p,i}, \overbar{d}'_{p,i}$ contribute $-(-(A\overbar{x})_i)=(A\overbar{x})_i$ to the potential of $q_i$. If $s=1$ then $\dec(\overbar{d}_{n,i})=-(A\overbar{x})_i$ and these neurons contribute $(A\overbar{x})_i-2(A\overbar{x})_i=-(A\overbar{x})_i$ to the potential of $q_i$.  

Similarly, as for the output neurons $\overbar{y}_i$, $\overbar{y'}_i$, if $s_i=0$, these neurons contribute $\dec(\widehat{y}_i)$ to the potential of $q_i$, and if $s_i=1$, these neurons contribute $-\dec(\widehat{y}_i)$. By choosing $\tau_2=\tau'+2$, the claim follows. 
\end{proof}
\begin{claim}
The output neurons encode the values $\widehat{y} + (1-2s)A\overbar{x}$ by round $\tau'=\tau_0 + O(\log \ell)$. Moreover, the output neurons continue to present the updated value until the next update presented in the input neurons. 
\end{claim}
\begin{proof}
By Claim~\ref{clm:sketch-potential} for every coordinate $i$, the potential value of $q_i$ is equal to the updated value $(\widehat{y} + (1-2s)A\overbar{x})_i$ 
by round $\tau_2 = \tau_0+ O(\log \ell)$. Therefore, it also holds that the potential value of $q'_i$ is equal to $-(\widehat{y} + (1-2s)A\overbar{x})_i$ by round $\tau_2$. 
Thus, by Lemma~\ref{clm:binary-potential} if $(\widehat{y} + (1-2s)A\overbar{x})_i \geq 0$ then the neurons $Q_i$ encode $(\widehat{y} + (1-2s)A\overbar{x})_i$ by some round $\tau_3 = \tau_2 + O(\log \ell)$ and are idle otherwise. Similarly, if $-(\widehat{y} + (1-2s)A\overbar{x})_i \geq 0$ then $Q'_i$ encodes $-(\widehat{y} + (1-2s)A\overbar{x})_i$ by round $\tau_3$. 
Note that either $Q_i$ or $Q'_i$ holds firing neurons but not both. 
In case $Q'_i$ contains firing neurons, $-(\widehat{y} + (1-2s)A\overbar{x})_i> 0$, and $q'_i$ fires as well (in such a case the sign neuron of the $i^{th}$ coordinate $s_i$ should be updated to $1$).
	 
We set the chain $C$ to be of size $\tau = O(\log \ell) > \tau_2 - \tau_0+2$. Since the first neurons in the delay chains that correspond to $q'_i,Q'_i, Q_i$ serve as $\AND$ gates of $c_\tau$ and the corresponding neurons in $q'_i,Q'_i, Q_i$, they begin to fire only after the neurons $q'_i,Q'_i, Q_i$ hold the correct values in round $\tau_0 +\tau+1$. Additionally, due to the inhibition of $c_{\tau+1}$, starting round $\tau_0+\tau+2$ all neurons in the network (including $y_i, s_i$) are idle except the delay chains which holds the updated values. We set all the delay chains connected to the output neurons to be of size $3$ and therefore in round $\tau_0 + \tau+3$ the neurons $y_i$, $s_i$ will be updated once with the correct value (after these neurons are already reset). Moreover, due to the self loops of the output neurons they will continue to present the updated value until the next update is presented.  
\end{proof}

\section{The Distinct Elements Problem}\label{sec:distinct}
In the distinct elements problem, given is a stream of integers $\mathcal{S}=\{x_1, \ldots, x_m\}$ where each $x_i \in [n]$. It is then required to maintain an estimate for the number of distinct elements in the stream. 
We start by stating the state-of-the-art bounds for this problem in the streaming model. 
\begin{fact}[Streaming Space Bounds, \cite{Blasiok18},\cite{LiW13}]\label{fc:dist-el}
For any $\epsilon, \delta \in (0,1)$, there is an $(1+\epsilon)$ approximation algorithm for distinct elements with success probability of $1-\delta$ using $O(1/\epsilon^2 \cdot\log 1/\delta+\log n)$ space. Moreover, the space bound is optimal.
\end{fact}
In this section, we provide a neural implementation for the well-known LogLog streaming algorithm by \cite{DurandF03,flajolet2007hyperloglog}. This algorithm obtains sub-optimal space, but its simplicity makes it much more applicable in the neural setting. 
\begin{lemma}\cite{DurandF03,flajolet2007hyperloglog} \label{lem:stream-distinct}
Given a data-stream of elements in $[n]$, there exists an $(1+\epsilon)$ approximation algorithm for the distinct elements problem using $O((1/\epsilon^2 \log \log n + \log n)\log(1/\delta))$ space, with probability of $1-\delta$. 
\end{lemma}
\begin{proof}[Proof Sketch]
We describe the high-level idea of the randomized LogLog algorithm under a constant success guarantee.
To provide a success guarantee of $1-\delta$, the same algorithm is repeated for $O(\log(1/\delta))$ times, thus inuring an overhead of $O(\log(1/\delta))$ factor in the space complexity.

The algorithm uses a pairwise independent hash function $h: [n] \to [2^{(2\log 1/\epsilon+3\log n)}]$ 
to map each input value $x_i \in [n]$ into a random string $h(x_i)$. The first $k=2\log 1/\epsilon$ bits of $h(x_i)$ are used in order to map $x_i$ into $2^k=1/\epsilon^2$ buckets $b_1,\ldots, b_{2^k}$. Let $\rho(x_i)$ be the number of leading zeros in the remaining $3\log n$ bits of $h(x_i)$. In each bucket $b_{\ell}$, the algorithm maintains the maximum value of $\rho(x_i)$ for every stream element $x_i$ that maps into the bucket $b_\ell$. Denote this maximum value by $N_{\ell}$. The estimate on the number of distinct elements is given by 
$\alpha \cdot 2^k \cdot 2^{\bar{N}}$ where $\bar{N}$ is the average of the $N_{\ell}$ values over the $2^k$ buckets $b_1,\ldots, b_\ell, \ldots, b_{2^k}$ for some constant $\alpha$.
\end{proof}
\vspace{-10pt}
\paragraph{A Neural Network for the Distinct Element Problem.}
We next turn to describe a neural implementation of the classical LogLog algorithm and prove Theorem~\ref{lem:disnct-elements}. 

\begin{definition}[SNN for Distinct Elements]
Given parameters $n$ and $\epsilon,\delta \in (0,1)$, an SNN network $\cN$ for the distinct elements problem has $n$ input neurons $\overbar{x}$,  and $\log n$ output neurons $\overbar{y}$. For every round $t$, let $f_1(t)$ 
be\footnote{We call it $f_1$ since the distinct elements problem computes the $F_1$ norm of the stream.} the number of distinct elements arrived by round $t$. For every fixed round $t$, it holds that by round $\tau(t)$, the output neurons $\overbar{y}$ encode the binary representation of an $(1 + \epsilon)$ approximation of $f_1(t)$ with probability $1-\delta$. 
\end{definition}

We describe the network construction based on the sequence of operations applied on the input layer. See Figure~\ref{fig:distinct} for an illustration. 
\paragraph{(1) Encoding the input in a binary Form.} The network contains $\log n$ neurons $x'_1, \ldots x'_{\log n}$ that represent the binary encoding of the presented element. The algorithm connects the input neurons $\overbar{x}$ to the neurons $\overbar{x}' = x'_1, \ldots , x'_{\log n}$ such that for every $i \in [n]$, $j \in [\log n]$ the edge weight $w(x_i,x'_j)=1$ if the $j^{th}$ bit in the binary representation of $i$ is equal to $1$ and $w(x_i,x'_j)=0$ otherwise. The bias values of these neurons are set to $b(x'_j)=1$ for every $j\in \{1, \ldots \log n \}$.

\paragraph{(2) Hashing.} The network contains a sub-network
$\mathcal{H}$ which implements a pairwise independent hash function $h: \{0,1\}^{\log n} \to \{0,1\}^{2\log (1/\epsilon)+3\log n}$ using Lemma~\ref{lem:hash}.
The input to the sub-network $\mathcal{H}$ are the neurons $\overbar{x}'$. 
Let $t=O(\log \log n)$ be an upper bound on the number of rounds required for the computation of the sub-network $\mathcal{H}$.
In order to maintain the persistence of the input $\overbar{x}'$ for $\Theta(t)$ rounds, the network contains a neural timer $\mathcal{NT}$ using Fact~\ref{fc:neural-counter}.

The output of $\mathcal{H}$ is denoted by $\overbar{h}_b$, $\overbar{h}_s$, where $\overbar{h}_b$ is of length $2\log (1/\epsilon)$, and $\overbar{h}_s$ is of length $3 \log n$. The output neurons $\overbar{h}_b$ will encode the bucket number the input is mapped to, and the vector $\overbar{h}_s$ holds a binary string of size $3\log n$. 
The neurons $\overbar{h}_b$, $\overbar{h}_s$ also have inhibitory copies denoted by $\overbar{h}'_b,\overbar{h}'_s$.

 \paragraph{(3) Computing the number of leading zeros.} In the next step the network computes the number of leading zeros in the hashed string $\overbar{h}_s$.
 For that purpose, we first connect the inhibitory neurons $\overbar{h}'_s$ to $\overbar{h}_s$ such that $\overbar{h}_s$ will contain a \emph{single} firing neuron corresponding to the leading one entry in the binary string $\overbar{h}_s$.
 This is done by connecting each inhibitory neuron $h'_{s,i}$ to the neurons $h_{s,1}, \ldots ,h_{s , i-1}$ with large negative weight. As a result, $\overbar{h}_s$ contains one firing neuron such that the neuron $h_{s,i}$ fires iff the number of leading zeros in the hashed string $\overbar{h}_s$ is $(3 \log n -i)$.
 Next, the number of leading zeros in $\overbar{h}_s$ is encoded into a binary form using a collection of $O(\log \log n)$ neurons denoted as $\overbar{z}$, which have incoming edges from $\overbar{h}_s$.

\paragraph{(4) Representing the $B=1/\epsilon^2$ buckets.} The buckets used in the LogLog algorithm are represented using $B$ sets of neurons $\overbar{b}_1 \ldots \overbar{b}_B$, each of cardinality $\log \log n$. To maintain the value stored in each bucket, the neurons $\overbar{b}_1 \ldots \overbar{b}_B$ have self-loops with large positives weights. Additionally, each set of neurons $\overbar{b}_i$ is connected to an inhibitor copy $\overbar{b}'_i$.
The invariant is that for all inputs seen so far that were mapped to bucket $i$, the firing state of $\overbar{b}_i$ will encode the maximum number of leading zeros among all the observed strings $\overbar{h}_s$.

In order to extract the index of the current bucket, the algorithm introduces $B$ excitatory neurons $a_1, \ldots a_B$, with incoming edges from the neurons $\overbar{h}_b, \overbar{h}'_b$ such that $a_i$ fires only if $\dec(\overbar{h}_b)=i$.

\paragraph{(5) Comparing the number of leading zeros with the value stored in the buckets.}
Let $j=\dec(\overbar{h}_b)$ be the decimal value encoded is the neurons $\overbar{h}_b$ at round $t$.
In the next step, our goal is to compare the number of leading zeros encoded in the neurons $\overbar{z}$ with the 
value stored in the $j^{th}$ bucket, encoded using the neurons $\overbar{b}_j$. 
In order to control the precise timing of the comparison, the network introduces a \emph{delay chain} of size $\tau= O(\log \log n)$ denoted as $C =\sigma_1, \ldots \sigma_\tau$. The first neuron in the chain $\sigma_1$ serves as an $\OR$ gate of the input neurons $\overbar{x}$, and for $i=2, \ldots ,\tau$ the neuron $\sigma_i$ has an incoming edge from $\sigma_{i-1}$ with weight $1$ and bias $1$.

The network then compares the value stored in the buckets $\overbar{b}_1, \ldots \overbar{b}_B$ with the current value stored in $\overbar{z}$ using $B$ comparison neurons $c_1, \ldots c_B$. Each comparison neuron $c_i$ will fire only if (i) $\sigma_{\tau}$ fired, ensuring the comparison occurs when $\overbar{z}$ holds the correct value, (ii) $a_i$ fired, indicating the input is mapped to bucket $i$, and (iii) $\dec(\overbar{z}) > \dec(\overbar{b}_i)$, indicating the current number of leading zeros is larger than the value stored in the bucket $\overbar{b}_i$.
This is done by setting $c_i$ to have incoming edges from $a_i$ and $\sigma_{\tau}$ with weight $w(a_i, c_i)=w(\sigma_\tau,c_i) = 10 \log n$, incoming edges from $\overbar{z}$ with weight $\dec(\overbar{z})$, negative incoming edges from $\overbar{b}'_i$ with weight $-\dec(\overbar{b}_i)$, and bias $b(c_i)=20 \log n+1$.

\paragraph{(6) Updating the relevant bucket.}
Once the comparison neuron $c_j$ fires, the goal is to copy the new value encoded in $\overbar{z}$ to the bucket $\overbar{b}_j$. First, the current value stored in the bucket is deleted as follows. For every $i \in [B]$ the neuron $c_i$ is connected to an inhibitory neuron $r_i$, and $r_i$ is connected to $\overbar{b}_i$  with a large negative edge weight. In order to update the value stored in the bucket \emph{after} the deletion,  the neuron $c_i$ is connected to a chain of two neurons $c^1_i$ and $c^2_i$, such that $c^1_i$ has an incoming edge from $c_i$, and $c^2_i$ has an incoming edge from $c^1_i$. 
Next, for every $k$ and $i$, the $k^{th}$ neuron in $\overbar{b}_i$ serves as an $\AND$ gate of $c_i^2$ and the $k^{th}$ neuron $z_k$, 

To avoid additional false updates, the neuron $c_i^1$ is connected to an inhibitory neuron $r_{i,2}$ that has negative outgoing edge weights to the chain $C$, the neurons $\overbar{z}$, and the input neurons $\overbar{x}$. We note that in a setting where there is a signaling neuron that fires upon the arrival of a new element, the inhibition of the input neurons can be avoided.
 
\paragraph{(7) Averaging.} All neurons $\overbar{b}_1, \ldots \overbar{b}_B$ are connected to an intermediate neuron $p$ such that the potential of $p$ is set to be $\log m \cdot \alpha \cdot \sum_{i=1}^{B}\dec(\overbar{b}_i)$, where $\alpha$ is a constant chosen according to the LogLog algorithm. The potential of $p$ is encoded by the output neurons $\overbar{y}$ using the $\POT(p)$ sub-network of Lemma~\ref{clm:binary-potential}. 

\paragraph{(8) Amplification of the success guarantee.} To amplify the success probability to $1-\delta$, the final network consists of $k=O(\log 1/\delta)$ copies of the basic sub-network described above. The final estimation is obtained by computing the median of the $k$ outputs of the sub-networks denoted as $\overbar{y}_1, \ldots, \overbar{y}_k$. This is done using a variation of the network for computing the maximum value by Maass~\cite{maass2000computational} as described in Fact~\ref{fc:max}. 

\begin{figure}[h!]
\begin{center}
\includegraphics[scale=0.5]{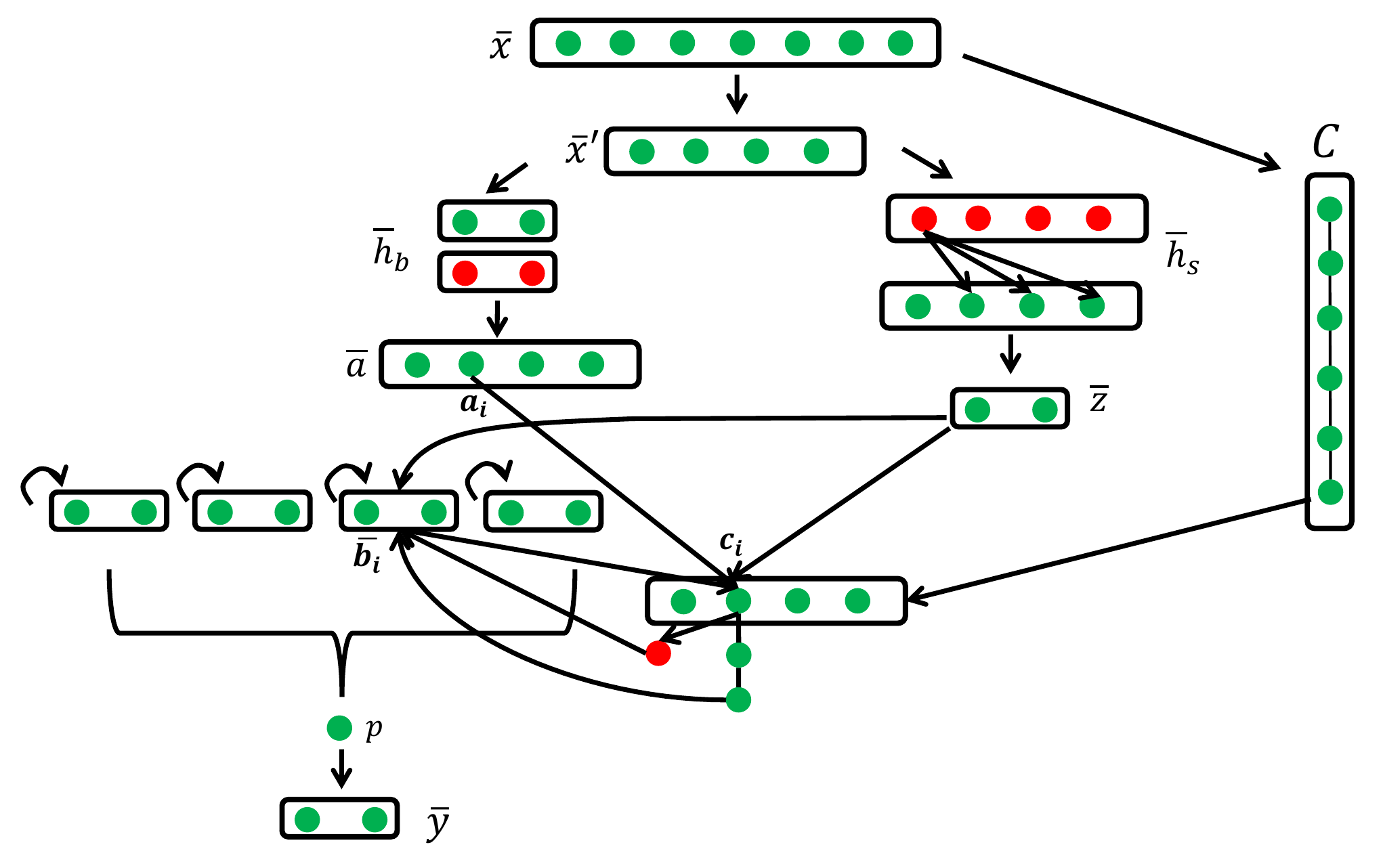}
\caption{\label{fig:distinct} \small{The distinct elements network. The red circles represent inhibitory neurons and the green circles represent excitatory neurons. First, the input element $x$ is encoded into the $\log n$ neurons $x'$ representing the value of $x$ in binary. Then $x'$ is hashed into two strings: (i) $\overbar{h}_b$ that encodes the bucket to which $x$ is mapped, and (ii) $\overbar{h}_s$ which is the $3\log n$-bit suffix of the hash value of $x$. The network then encodes the number of leading zeros in the hash string $\overbar{h}_s$ into a binary vector $\overbar{z}$, and extracts the bucket index into a unit vector $\overbar{a}$. The current maximum leading zeros in each bucket is stored in the counters on the left ($\overbar{b}_1, \ldots \overbar{b}_B$). Next, the network compares between the number of leading zeros in the current string and the value stored in the relevant bucket using the comparison neurons $c_1, \ldots c_{B}$. The chain $C$ is responsible for scheduling the comparison to occur only after the updated value has been computed. If needed, the corresponding bucket $\overbar{b}_i$ is updated with the value encoded in $\overbar{z}$ using the neurons connected to $c_i$.}}
\end{center}
\end{figure}
We are now ready to analyze the correctness of the construction, and by that complete the proof of Theorem \ref{lem:disnct-elements}.

\paragraph{Proof of Theorem \ref{lem:disnct-elements}.}
We show that the proposed network implements the LogLog algorithm of~\cite{DurandF03}. 
Given an input $\overbar{x}$ representing an element $i \in [n]$ introduced in round $\tau_0$, in round $\tau_0+1$ the neurons $\overbar{x}'$ hold the binary encoding of $i$. 
Moreover, due to the neural timer $\mathcal{NT}$ connected to $\overbar{x}'$, we can assume that $\overbar{x}'$ keeps presenting the value $i$ for $O(\log \log n)$ rounds. 

By Lemma~\ref{lem:hash}, the neurons $\overbar{h}_s$ encode the output of a pairwise independent hash functions $h_1: [n] \rightarrow [n^3]$, and $\overbar{h}_b$ encodes the output of a pairwise independent hash functions $h_2: [n] \rightarrow [1/\epsilon^2]$ by round $\tau_1 = \tau_0 + O(\log \log n)$.
We next observe that the neurons $\overbar{z}$ encode the number of leading zeros in $h_1(\overbar{x})$ starting round $\tau_1+2$. 
\begin{observation} \label{obs:leading}
Starting round $\tau_1+2$ it holds that $\dec(\overbar{z})$ encode the number leading zeros in $h_s(\overbar{x})$.
\end{observation}
\begin{proof}
By Lemma~\ref{lem:hash}, the neurons $\overbar{h}_s$ encode $h_1(\overbar{x})$ by round $\tau_1$. Due to the inhibition of $\overbar{h}'_s$ in Step 3, starting round $\tau_1+1$ the only neuron firing in $\overbar{h}_s$ is the leading one in the binary representation of $h_1(\overbar{x})$. Hence, if $\overbar{h}'_{s,i}$ fires in round $\tau+1$, then the number of leading zeros in $h_1(\overbar{x})$ is $|\overbar{h}_s|-i = 3\log n - i$. Due to the edges from $\overbar{h}_s$ to $\overbar{z}$, starting round $\tau_1+2$ the neurons $\overbar{z}$ holds the binary encoding of the number of  leading zeros in $h_1(\overbar{x})$.
\end{proof}
Let $j= h_2(\overbar{x})= \dec(\overbar{h}_b)$ be the bucket to which the input $x$ is mapped to in round $\tau_1$. 
We first claim that for every bucket $j' \neq j$, the neurons $\overbar{b}_{j'}$ do not change their values (from round $\tau_0$ and as long as the input has not changed). 
Starting at round $\tau_1+1$, the neuron $a_j$ -- corresponding to bucket $j$ -- fires, where for every $j' \neq j$, the neuron $a_{j'}$ is idle. Thus, the comparison neuron $c_{j'}$ does not fire and therefore for every index $j' \neq j$ the value stored in bucket $\overbar{b}_j$ does not change.

As for $\overbar{b}_{j}$, let $v_j(\tau)$ be the value stored in bucket $\overbar{b}_j$ in round $\tau$. 
We claim that if the number of leading zeros in $h_2(\overbar{x})$, denoted as $v_0$, is larger than $v_j(\tau_0)$,  then for $\tau_2 = \tau_0 + O(\log \log n)$ it holds that $v_j(\tau_2)= v_s$. 
We note that if $v_0 \leq v_j(\tau_0)$ 
by Observation~\ref{obs:leading} starting round $\tau_1+2$ also $\dec(\overbar{z}) \leq v_j(\tau_0)$. Setting $\tau > \tau_1+2 -\tau_0$, the comparison neuron $c_j$ will not fire and therefore the value stored in $\overbar{b}_j$ does not change.

\begin{claim} \label{clm:dist-new}
If $v_0 > v_j(\tau_0)$,  then for round $\tau_2 = \tau_0 + O(\log \log n)$ it holds that $v_j(\tau_2)=v_0$. 
Moreover, the value stored in $\overbar{b}_j$ will not change until the next update is presented.
\end{claim}
\begin{proof}
The index neuron $a_j$ starts firing by round $\tau_1+1$ due to the incoming edges from $\overbar{h}_b$.
We set the parameter $\tau$ such that $\tau > \tau_1 -\tau_0 + 3$. 
Thus, in round $\tau' = \tau_0+\tau+1$ both $\sigma_{\tau}$ and $a_j$ fires. 
Since $\tau' > \tau_1+2$, by Observation~\ref{obs:leading} in round $\tau'$ it also holds that $\dec(\overbar{z})=v_0$ and $\dec(\overbar{z}) > v_j(\tau_0)$.  We conclude that the comparison neuron $c_j$ fires in round $\tau'+1$. 
Due to the inhibitor $r_i$ all neurons in $\overbar{b}_j$ are idle starting at round $\tau'+3$, and due to the edges from the neuron $c^2_{i}$, in round $\tau'+4$ the neurons $\overbar{b}_{j}$ will obtain the value $v_0$ encoded using the neurons $\overbar{z}$. 

Additionally, due to the inhibitor $r_{i,2}$, starting round $\tau'+4$, the neurons $\overbar{z}$, $C$ and $\overbar{x}$ are idle. Therefore, the comparison neuron $c_j$ will not fire until the next input is presented and no additional update will be performed. The claim follows for $\tau_2=\tau'+4$. 
\end{proof}
Combining Lemma~\ref{lem:stream-distinct}, Claim~\ref{clm:dist-new} and Steps (7,8), the upper bound of Theorem~\ref{lem:disnct-elements} is established.

\paragraph{Lower Bound.} Finally, we show that the space-bound of Theorem \ref{lem:disnct-elements} is nearly optimal by using a reduction from communication complexity. 
\begin{lemma}\label{lem:distel}[Neural Lower Bound]
Any SNN for maintaining a $(1+\epsilon)$ approximation for the number of distinct elements with constant probability requires $\Omega(1/\epsilon^2+\log n)$ neurons. 
\end{lemma}
\begin{proof}
The lower bound is based on a reduction from communication complexity in the same manner as was shown for the streaming setting. Specifically, for the streaming setting Woodruff \cite{woodruff2004optimal} showed a lower bound of $\Omega(1/\epsilon^2)$ space for the $(1+\epsilon)$ Distinct Elements problem with $\delta=O(1)$. This was shown via a reduction from the Gap-Hamming problem in the public-coin two-party model. 
In our context, we use a similar reduction in order to show a lower bound of $\Omega(1/\epsilon^2)$ on the number of neurons in an SNN network. Let $\mathcal{N}$ be a network for maintaining an $(1+\epsilon)$ estimate for the number of distinct elements with constant probability using $s$ non-input neurons. We use this network to provide a one-way communication complexity protocol between Alice and Bob. Since the random coins are public, both Alice and Bob can compute the network $\mathcal{N}$ (i.e., with the random edge weights). Alice simulates her input items over the network $\mathcal{N}$ by feeding them as input to the network (for a fixed number of rounds). She then sends to Bob the firing states of the non-input neurons in $\mathcal{N}$. This allows Bob to continue with the network simulation by feeding it its input items. The correctness follows by the correctness of the SNN network $\mathcal{N}$ and the communication complexity lower bound.

To show a lower bound of $\Omega(\log n)$ we will use the reduction to the Disjointedness problem in the communication-complexity setting by Alon at el. \cite{alon1999space}. This reduction as well works in the two-party model with public-coins, which allows Alice and Bob to compute the same network $\mathcal{N}$ and simulate its operation over their input items in the same manner as above. 
\end{proof}

\vspace{-10pt}
\section{Median Approximation}\vspace{-5pt}
Before presenting the neural computation of the approximate median, we describe the neural implementation of the Count-Min Sketch and prove Theorem~\ref{cor:min-sketch}. 
\vspace{-10pt}\subsection{A Neural Implementation of Count-Min Sketch} \label{sec:count-min}
We follow the streaming implementation of Count-Min by \cite{cormode2005improved} described as follows. 
The algorithm maintains a data structure that consists of $\ell= O(\log 1/\delta)$ hash tables $T_1, \ldots T_\ell$, each with $b=O(1/\epsilon)$ bins, and each table $T_i$ is indexed using a different pairwise independent hash function $h_i$ (i.e., the output domain of $h_i$ is $\{0,1\}^{\log b}$). The operation $\inc(x)$ increases the value in each bin $T_i[h_i(\overbar{x})]$ for every $i \in [\ell]$. The $\countt(x)$ operation returns the value $\min_{i \in [\ell]} T_i[h_i(\overbar{x})]$. 
\begin{fact}[\cite{cormode2005improved}] \label{fact:count-min}
$\Pr[\countt(x) \notin (f(x),f(x)+O(m/b))] \leq 1/2^{\Omega(\ell)}$ where $f(x)$ is actual frequency of $x$ in the stream of length $m$. 
\end{fact}
\begin{definition}[Neural Count-Min Sketch]
Given parameters $\epsilon,\delta \in (0,1)$, a neural Count-Min sketch network $\mathcal{N}_{\epsilon,\delta}$ has an input layer of $\log n+1$ neurons denoted as $a, x_1, \ldots x_{\log n}$, an output layer of $\log m$ neurons $y_1, \ldots y_{\log m}$, and a set of $s$ auxiliary neurons. The neurons $\overbar{x}=(x_1, \ldots x_{\log n})$ encode the binary representation of an element $x \in [n]$ and the neuron $a$ indicates whether this is an $\inc$ or $\countt$ operation, where $a=1$ indicates an $\inc$ operation. 
For every fixed input value $\overbar{x}$ presented at round $t$ and $a=0$ (i.e., a $\countt$ operation), 
let $\mathcal{N}_{\epsilon,\delta}(\overbar{x})$ be the value encoded in binary by the output layer $y_1, \ldots y_{\log m}$ in round $t+\tau_{n,m}$.  It holds that $\Pr[\mathcal{N}_{\epsilon,\delta}(\overbar{x})\notin (f(x),f(x)+O(\epsilon m'))]\leq \delta$, where $m' \leq m$ is the stream length by round $t$ and $f(x)$ is the current frequency of $x$.
\end{definition}
We first describe the network construction to support the $\inc(x)$ operation. Then we explain the remaining network details for implementing a $\countt(x)$ operation.
 
\vspace{-10pt}
\paragraph{Supporting $\inc(x)$ operation.} 
The network contains $\ell=O(\log 1 / \delta)$ sub-networks
$\mathcal{H}^1_{n,b},\ldots, \mathcal{H}^\ell_{n,b}$ each implements a pairwise independent hash function $h_i: \{0,1\}^{\log n} \to \{0,1\}^{\log b}$ using Lemma~\ref{lem:hash}.
The output vector of each network $\mathcal{H}^i_{n,b}$ is denoted by $\overbar{h}_i$ for every $i \in \{1,\ldots,\ell\}$. Every $\overbar{h}_i$ has an inhibitory copy $\overbar{h}'_i$.

For each sub-networks $\mathcal{H}^i_{n,b}$, and for every value $j \in \{1,\ldots, b\}$, the network contains a counter sub-network that counts the number of data-items $x$ in the stream that satisfies $h_i(\overbar{x})=j$. 

Every counter network is implemented by a neural counter network from Fact \ref{fc:neural-counter} with time parameter $t=m$. Let $\cC_{i,1}, \ldots, \cC_{i,b}$ be the neural counter networks corresponding to the $i^{th}$ hash network $\mathcal{H}^i_{n,b}$. The counter $\cC_{i,j}$ is updated based on the values of the output neurons $\overbar{h}_i$ as follows. 
For every counter $\cC_{i,j}$ the network contains an index neuron $c_{i,j}$ with input from $\overbar{h}_i$ and $\overbar{h'}_i$ which fires only if\footnote{For implementation reasons, verifying that $\dec(\overbar{h}_i)=j$ requires input from both $\overbar{h}_i$ and $\overbar{h'}_i$.} $\dec(\overbar{h}_i)=j$. The input to the counter $\cC_{i,j}$ denoted as $e_{i,j}$ is an $\AND$ gate of the input neuron $a$ and the index neuron $c_{i,j}$, firing in $\inc(x)$ operations where $h_i(\overbar{x})=j$.
To make sure the counter is incremented once per $\inc(x)$ operation, the network contains an inhibitory neuron denoted as $e'_{i,j}$ which has the same incoming edges and weights as $e_{i,j}$, that inhibits the neurons $e_{i,j}$, $c_{i,j}$ and $a$. This guarantees that $e_{i,j}$ would be active for exactly \emph{one} round per $\inc(x)$ operation.
\paragraph{Supporting $\countt(x)$ operation.} 
To support a $\countt(x)$ operation, for each counter $\cC_{i,j}$, the network includes $\log m$ neurons $\overbar{s}_{i,j}=s^1_{i,j}, \ldots s^{\log m}_{i,j}$ which hold the value stored in the counter $\cC_{i,j}$ such that $h_i(\overbar{x})=j$. Each neuron $s_{i,j}^k$ is an $\AND$ gate of the index neuron $c_{i,j}$ and the $j^{th}$ output neuron of $\cC_{i,j}$.  In addition, for every $i \in \{1,\ldots, \ell\}$ there are $\log m$ neurons $\overbar{g}_i=g_{i,1}, \ldots g_{i,\log m}$ where the $j^{th}$ neuron $g_{i,j}$ is an $\OR$ gate of all the $j^{th}$ neurons of the vectors $\overbar{s}_{i,1}, \ldots, \overbar{s}_{i,b}$. As a result, $\overbar{g}_i$ encodes the value stored in $h_i(\overbar{x})$. Finally, the output value is set to be the \emph{minimum value} of $\dec(\overbar{g}_1), \ldots ,\dec(\overbar{g}_\ell)$ using the minimum computation network of \cite{maass2000computational}. 

\paragraph{Correctness.} Our goal is to show the proposed network simulates the Count-Min sketch data structure of~\cite{cormode2005improved}. 
Let $(x,a)$ be an input introduced in round $\tau_0$. We start by showing the correctness of an $\inc(x)$ operation (i.e., when $a=1$). Specifically, we show that the counters $\cC_{i,h_i(\overbar{x})}$ are incremented by one, and the remaining counters $\cC_{i,j}$ for $j \neq h_i(\overbar{x})$ are unmodified. 
\begin{claim} \label{clm:counter}
For every $i \in \{1,\ldots, \ell\}$, the value encoded in the output neurons of the counter $\cC_{i,h_i(\overbar{x})}$ is
increased by one by round $\tau_2 = \tau_0 + O(\log m +\log \log n)$, and for every $j \neq h_i(\overbar{x})$ the output of the counter $\cC_{i,j}$ is not incremented. The increment to the $\cC_{i,h_i(\overbar{x})}$ counters occur only once.
\end{claim}
\begin{proof}
By Lemma~\ref{lem:hash}, the neurons $\overbar{h}_1, \ldots \overbar{h}_\ell$ encodes the values $h_1(\overbar{x}), \ldots h_\ell(\overbar{x})$ by round $\tau_1 = \tau_0 + O(\log \log n)$. 
Thus, for every $i \in \{1,\ldots, \ell\}$ starting round $\tau_1+1$, the index neuron $c_{i,j}$ fires iff $j=h_i(\overbar{x})$. Since $a=1$, starting round $\tau_0$, the $\AND$ gate $e_{i,h_i(\overbar{x})}$ fires in round $\tau_1+2$ (and $e_{i,j}$ are idle for $j \neq h_i(\overbar{x})$). Due to the inhibitor copy of $e_{i,j}$, in round $\tau_1+3$ the neuron $e_{i,j}$ does not fire. Moreover, since the inhibitor also inhibits $a$ nd $c_{i,j}$, the neuron $e_{i,j}$ is idle until a new $\inc$ operation is presented. Hence, the input neuron $e_{i,j}$ of the counter network $\cC_{i,j}$ fires exactly once, and therefore the counter is incremented once, as desired. By Fact \ref{fc:neural-counter}, it holds that the output neurons of $\cC_{i,j}$ hold the correct count by round $\tau_2 = \tau_1 + O(\log m)$.
\end{proof}

We proceed with showing the correctness for the $\countt(x)$ operation (i.e., when $a=1$). 
For every $i \in \{1,\ldots, \ell\}$ and $j \in \{1,\ldots, b\}$, let $x_{i,j}(\tau)$ be the (decimal) value encoded by the output neurons of the counter $\cC_{i,j}$ in round $\tau$. We next show that by round $\tau_0+O(\log m + \log \log n)$, the output neurons $\overbar{y}$ of the Count-Min sketch encode the minimum value in $x_{1, h_1(\overbar{x})}(\tau_0), \ldots, x_{\ell, h_\ell(\overbar{x})}(\tau_0)$.
\begin{claim} \label{clm:min-output}
The output neurons $\overbar{y}$ encode the value $\min_{i \in \{1,\ldots, \ell\}} x_{i,h_i(\overbar{x})}(\tau_0)$ by round $\tau_2 = \tau_0 + O(\log m+\log \log n)$.  
\end{claim}
\begin{proof}
By Lemma~\ref{lem:hash}, the neurons $\overbar{h}_1, \ldots \overbar{h}_\ell$ encodes the values $h_1(\overbar{x}), \ldots h_\ell(\overbar{x})$ by round $\tau_1 = \tau_0 + O(\log \log n)$. In addition, for every $i \in \{1,\ldots, \ell\}$, the neuron $c_{i,j}$ fires starting round $\tau_1+1$ iff $j = h_{i}(\overbar{x})$. Since $a=0$, starting round $\tau_0$ the counter $\cC_{i,j}$ is not incremented and $c_{i,j}$ is not inhibited by the inhibitor copy of $e_{i,j}$.
Thus combined with the persistence assumption of every input, we conclude that the vector $\overbar{s}_{i,h_i(\overbar{x})}$ encodes $x_{i, h_i(\overbar{x})}(\tau_0)$ starting round $\tau_1+2$. In addition, for every $j \neq h_i(\overbar{x})$ all the neurons of $\overbar{s}_{i,j}$ are idle. Therefore, the neurons $\overbar{g}_i$ encode $x_{i, h_i(\overbar{x})}(\tau_0)$ starting round $\tau_1+3$. The claim follows from fact~\ref{fc:max}. 
\end{proof}
The Theorem follows by combining Claim~\ref{clm:counter}, Claim~\ref{clm:min-output} and Fact~\ref{fact:count-min}.

\subsection{Neural Computation of the Approximate Median}
In this section, we present our main technically involved algorithmic result for computing an estimate for the median of the data-stream. 
\begin{definition}[Approximate Median]
Given $\epsilon, \delta \in (0,1)$ and a stream $\mathcal{S} =\{ x_1, x_2, \ldots x_m\}$ with each $x_i \in [n]$, in the \emph{approximate median} problem, it is required to output an element $x_j \in \mathcal{S}$ whose rank is $m/2 \pm \epsilon m$ with probability at least $1-\delta$.
\end{definition}
For ease of notation, assume that $n$ is power of $2$.
Our neural solution is based on the streaming algorithm of \cite{cormode2005improved}, that uses $\widetilde{O}(1/\epsilon)$ space. Up to the logarithmic terms, this space-bound is known to be optimal \cite{karnin2016optimal}.
\begin{fact}[Theorem 5 \cite{cormode2005improved}] \label{fct:median}
For every $\epsilon, \delta \in (0,1)$, there exists a randomized streaming algorithm for computing the $\epsilon$-approximate median with probability $1-\delta$ and $\widetilde{O}(1/\epsilon)$ space.
\end{fact}
We start by providing a high-level exposition of this streaming algorithm, and then explain its implementation in the neural setting. The latter turns out to be quite involved, yet demonstrating the expressive power of SNN networks. 

\paragraph{A high-level description of the streaming algorithm.} The algorithm is based on applying a binary search over \emph{range queries} which, roughly speaking, compute the frequency of the elements in a given range. 
\begin{definition}[Range Queries] 
Given a data-stream of numbers $\mathcal{S}= \{x_1, \ldots , x_m\}$ with each $x_i \in [n]$, a range query receives a range of number $[a,b] \subseteq [1,n]$ and returns the frequency of the items $\{a,a+1, \ldots , b\}$ in the stream $\mathcal{S}$.
\end{definition}
To support range queries with small space, the algorithm maintains $\log n$ data structures of Count-Min sketch, for each of the $\log n$ \emph{dyadic intervals} of $[n]$. 
\begin{definition}[Dyadic Intervals]
The \emph{dyadic intervals} of the set $[n]$ are a collection of $\log n$ partitions of $n$,  $\mathcal{I}_1 \ldots ,\mathcal{I}_{\log n}$ such that
\begin{align*}
\mathcal{I}_0 &= \{\{1\}, \{2\} , \{3\}, \ldots , \{n\}\} \\
\mathcal{I}_1 &= \{\{1,2\}, \{3,4\} , \{5,6\}, \ldots , \{n-1, n\}\} \\
\mathcal{I}_2 &= \{\{1,2,3,4\}, \{5,6,7,8\} , \ldots , \{n-3,n-2, n-1, n\}\} \\
\dots \\
\mathcal{I}_{\log n} &= \{\{1,2 \ldots  n\}\}
\end{align*} 
\end{definition} 
Note that every range $[i,j] \subseteq [n]$ can be written as a union of at most $\log n$ sets from the dyadic intervals. Hence, by introducing $\log n$ Count-Min data structures with parameters $\delta' = \log (\log n/\delta)$ and $\epsilon' = \epsilon / \log n$ for dyadic-intervals of $[n]$, we can answer range queries within an additive error of $m \cdot \epsilon$ with probability $1- \delta$. 
The approximated median is obtained by employing a Binary search over the range queries \footnote{The same algorithm can be applied for any quantile estimation.}. 
\begin{definition}[SNN for the Approximate Median Problem]
Given two integers $n,m$ and additional parameters $\epsilon, \delta \in (0,1)$, an approximate-median network $\mathcal{N}_{n,m}$ has an input layer of $n + 1$ neurons, an output layer of $\log n$ neurons and a set of $s$ auxiliary neurons. The input neurons are denoted as $(a, x_1, \ldots, x_{n})$ where the neuron $a$ indicates whether this is a median query or an insertion operation.
When the input layer represents a median query, the neuron $a$ fires and the neurons $x_1, \ldots x_{n}$  are idle. 
For every round $t$, let $\mathcal{S}_t =  \{a_1, a_2, \ldots a_{t}\}$ be the data-stream presented as input to the network by round $t$. For any median-query presented in round $t$, by round $t+\tau_{n,m}$ the output layer encodes an element $y \in \mathcal{S}_t$ whose rank in $\mathcal{S}_t$ is $t/2\pm \epsilon t$ with probability at least $1-\delta$. 
\end{definition}

\paragraph{The challenge:} The crux of the streaming algorithm is based on a binary search over range queries. A-priori, it is unclear how to implement such a search using a poly-logarithmic number of neurons. Specifically, the (implicit) decision tree that governs the binary search has a linear size. Since the neural network (unlike the streaming algorithm) has to hard-wire the algorithm description, the explicit encoding of the search tree leads to a linear space solution. Our key contribution is in showing a succinct network construction that simulates the binary search of the streaming algorithm using a nearly matching space bound.

\begin{figure}[ht!]
\begin{center}
\includegraphics[scale=0.4]{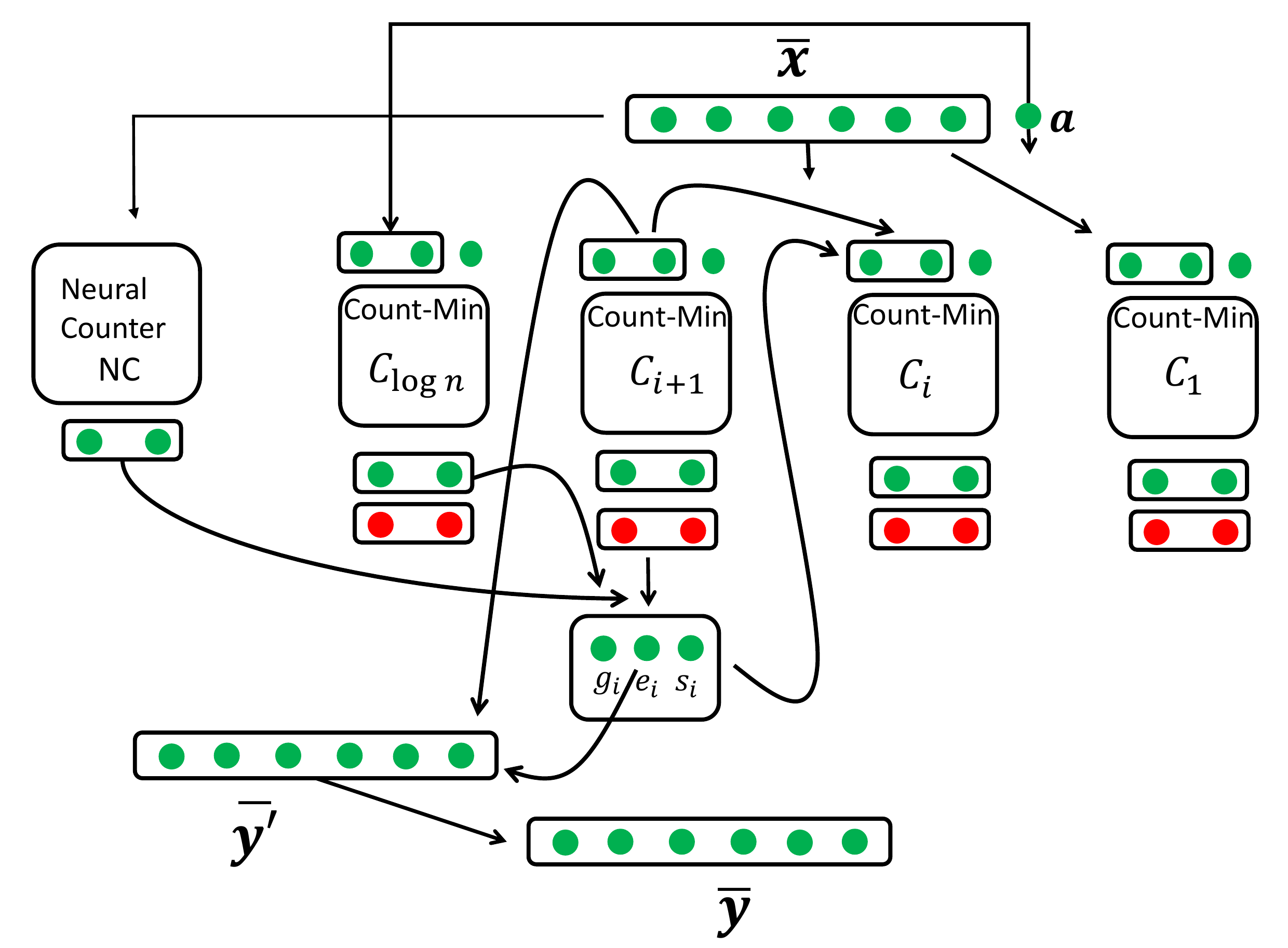}
\caption{ \label{fig:approx-median} \small{A high-level illustration of the approximate median network. The green circles represent excitatory neurons and the red circles represent inhibitory neurons. The input is connected to $\log n$ Count-Min networks, that count the frequencies of the dyadic intervals of $[n]$. On the left, the neural counter module $NC$ counts the total number stream elements. The neurons $g_i,s_i,e_i$ guide the binary search implemented by the network. Once a median is detected, its value is copied to the output neurons using the additional neurons $\overbar{y}'_i$. 
}}
\end{center}
\end{figure} 
\paragraph{Network description.} We next provide a description of the network. Recall that the type of the operation is represented by the input neuron $a$, where $a=1$ represents a median query.
To avoid cumbersome notation, we assume that $n$ is a power of $2$. 
\\ \\
\noindent\textbf{Supporting an insertion operation.} 
In the high level, the network contains 3 parts (1) a set of $\log n$ neurons that encode the inserted element in its binary form, (2) a neural counter that counts the length of the current stream, and (3) $\log n$ Count-Min sketch sub-networks that maintain the frequencies of the $\log n$ dyadic intervals of $[n]$.
\begin{enumerate}
\item The $n$-length input vector $\overbar{x}$ is connected to $\log n$ neurons $\overbar{x}' = (x'_1, \ldots , x'_{\log n})$ such that $\overbar{x}'$ encodes the binary representation\footnote{As discussed in the introduction our solution supports both types of input formats: $\log n$-bits of the binary representation or an $n$-length vector with one active entry.} of the element presented in the input neurons $\overbar{x}$.
For every $i \in [n]$, $j \in [\log n]$ if the $j^{th}$ bit in the binary representation of $i$ is equal to $1$ then $w(x_i,x'_j)=1$, and $w(x_i,x'_j)=0$ otherwise. The bias of every $x'_i$ is set to $b(x'_j)=1$.
\item  The network contains a counter sub-network  $\mathcal{NC}_m$ that counts the number of data-items inserted so far. 
The counter network is implemented by a neural counter network from Fact \ref{fc:neural-counter} with time parameter $t=m$. The input neuron to the  $\mathcal{NC}_m$ sub-network denoted as $a'$ is an $\OR$ gate of the input neurons $\overbar{x}$. The output neurons of $\mathcal{NC}_m$ are denoted as $\overbar{o}=(o_1 , \ldots, o_{\log m})$. The network also contain inhibitory copy of the vector $\overbar{o}$ denoted by $\overbar{o}'$.

To make sure the counter is incremented once per insertion operation, the network contains an inhibitory copy of $a'$ denoted as $r'$,
which inhibits $a'$ and the neurons $\overbar{x}$ using large negative weights. As a result, the input neuron $a'$ will be active for exactly \emph{one} round per insertion operation.

\item  The network contains $\log n$ sub-networks $\cC_1, \ldots \cC_{\log n}$ each implements a Count-Min sketch with parameters $n,m$ and $\epsilon'=O(\epsilon/\log n)$, $\delta'=O(\delta/\log n)$ using Theorem~\ref{cor:min-sketch}. 
For each Count-Min sketch sub-network $\cC_i$, let $\overbar{z}_i = (z_{i,1}, \ldots z_{i, \log n})$ and $b_i$ be its input layer, where the neuron $b_i$ indicates whether the operation is $\inc(x)$ or $\countt(x)$.
The neuron $b_i$ is an $\OR$ gate of the neurons in $\overbar{x}$. 

For $i \in \{1,\ldots, \log n\}$, the input neurons $\overbar{z}_i$ are connected to the binary representation of the input $\overbar{x}'$ in the following manner. For every $j \geq i$, the neuron $x'_j$ is connected to the neuron $z_{i,j}$ with large positive weight. For every $j<i$ the neuron $z_{i,j}$ serves as an OR gate of the neurons of $\overbar{x}'$.
In addition, the neurons $b_i$, $\overbar{z}_i$ are equipped with self-loops. The Count-Min sketch sub-networks are then modified such that these neurons will be inhibited once the computation is complete (by the inhibitory neurons $e'_{i,j}$ of each sub-network respectively). 
\end{enumerate}

\noindent\textbf{Supporting a Median Query.}
Given a median query, the network computes the approximate median by employing at most $\log n$ steps of binary search. In every step\footnote{It is convenient to count the steps in a backward manner, as in the $i^{th}$ step the network will access the $i^{th}$ Min-Sketch module $\cC_i$.}  $i \in \{\log n, \ldots, 1\}$, the network obtains a current candidate for the median denoted by $\chi_i$. Initially, $\chi_{\log n}=n/2$. Each $\chi_i$ would be provided as input for the $i^{th}$ Count-Min sketch $\cC_i$. The output neurons of $\cC_i$ would then define the next candidate $\chi_{i-1}$. Specifically, depending on the rank estimation of $\chi_i$, the network defines the new search range. The width of the search range would be cut by a factor $2$ in every step $i$. Consequently, the algorithm will be using the Count-Min sketch $\cC_{i-1}$ which is defined over a partitioning $\mathcal{I}_{i-1}$ in which each set is smaller by factor $2$ compared to $\mathcal{I}_{i}$.
\begin{itemize}
\item [(1)] For every $i \in \{1,\ldots,\log n\}$ the network contains an additional Count-Min sub-network $\cC'_{i}$ which counts the frequencies of the data-elements (similar to $\cC_1$). These additional Count-Min sub-networks will be useful in a scenario where for the median item $j^*$, the frequency of the range  $[1,j^*]$ is larger than half, and the frequency of $[1, j^* - 1]$ is too small. This special case would be handled using the $\cC'_{i}$ sub-networks. 

For every sub-network $\cC'_{i}$ with input $\overbar{z}'_{i}, b'_{i}$, the neuron $b'_{i}$ serves as an $\OR$ gate of the neurons in $\overbar{x}$. As for $\overbar{z}'_{i}$, each neuron $z'_{i,j}$ serves as an OR gate of the $j^{th}$ neuron of $\overbar{x}'$ and the $j^{th}$ neuron of $\overbar{z}_i$. Hence, for insertion operations, the sub-network $\cC'_{i}$ is equivalent to $\cC_1$. Additionally, the neurons $\overbar{z}'_{i}, b'_{i}$ are equipped with self-loops. The Count-Min sketch sub-networks are then modified such that these neurons will be inhibited once the computation is complete.

\item [(2)] For every $i \in \{1,\ldots,\log n\}$ the network contains three \emph{comparison} neurons $s_i,g_i,e_i$ (corresponding to \emph{smaller, greater} or \emph{equal}). These neurons receive their input from the output neurons of the counters $\cC_{\log n }, \ldots, \cC_{i}$, $\cC'_{\log n }, \ldots, \cC'_{i}$, and the output of the neural counter $\overbar{o}, \overbar{o}'$. 
Let $\chi_i=\dec(\overbar{z}_i)$ be the median candidate at phase $i$ of the binary-search. The firing states of the comparison neurons are determined as follows. The neuron $g_i$ would fire if frequency estimation of $[1,\chi_i]$ is greater than $m'/2+ \epsilon/2 m'$. The neuron $s_i$ would fire if frequency estimation of $[1,\chi_i]$ is smaller than $m'/2-\epsilon/2 m'$. 
Finally, $e_i$ would fire if the frequency estimation of $[1,\chi_i]$ is in the range $(m'/2 - \epsilon/2 m', m'/2 + \epsilon/2)$.

Denote the output neurons of $\cC_i$ by $\overbar{f}_i$ and the output neurons of $\cC'_i$ by $\overbar{f}'_i$. 
As we will see, the frequency of the range $[1,\chi_i]$ will be decoded by the output neurons of $\cC_{\log n}, \ldots, \cC_{i}$ as $\sum_{j=(i-1)}^{\log n} \dec(\overbar{f}_{j})$.
Since the output neurons $\overbar{f_i}$ and $\overbar{f}'_i$ are \emph{excitatory} (i.e. may only have non-negative outgoing edges), in order for the comparison neurons to fire as desired, for every $i \in [\log n]$ the network introduces inhibitory copies of $\overbar{f}_i$ and $\overbar{f}'_i$ denoted as $\overbar{\phi}_i$ and $\overbar{\phi}'_i$ respectively with outgoing edges to the comparison neurons $g_i,s_i$. 

We set the incoming edges of the neuron $g_i$ such that
$g_{i}$ fires if $\sum_{j=i}^{\log n} \dec(\overbar{f}_{j})- \dec(\overbar{f}'_{i}) > \lfloor \dec(\overbar{o})/2 \rfloor + \epsilon/2\cdot  \dec(\overbar{o})$. 
Similarly, the neuron $s_{i}$ fires if $\sum_{j=i}^{\log n} \dec(\overbar{f}_{j}) -1 <  \lfloor \dec(\overbar{o})/2 \rfloor - \epsilon/2 \cdot \dec(\overbar{o})$. 
Regarding the equality neuron $e_{i}$, it serves as an $\AND$ gate of two intermediate neurons $e_{i,1}$, $e_{i,2}$ such that $e_{i,1}$ fires if
$\sum_{j=i}^{\log n} \dec(\overbar{f}_{j}) \geq  \lfloor \dec(\overbar{o})/2 \rfloor  - \epsilon/2 \cdot \dec(\overbar{o}) $ and $e_{i,2}$ fires if $\sum_{j=i}^{\log n} \dec(\overbar{f}_{j})-\dec(\overbar{f}'_{i})+1 \leq \lfloor \dec(\overbar{o})/2 \rfloor + \epsilon/2 \cdot \dec(\overbar{o})$.

The neuron $g_{i}$ is also connected to an inhibitor $g'_{i}$ which inhibits $\overbar{f}_{i}$ and $\overbar{f}'_{i}$ with large negative weight.
The inhibition of $\overbar{f}_{i}$, $\overbar{f}'_{i}$ allows us to maintain the invariant that $\sum_{j=(i-1)}^{\log n} \dec(\overbar{f}_{j})$ will hold the frequency estimation of $[1,\chi_{i-1}]$ in the next phase when considering $\chi_{i-1}$.

\item [(3)] The network is augmented with a \emph{timing chain} $T$ which schedules the updates of the neurons $\overbar{z}_i$ with the median candidate $\chi_i$. This update should be carefully coordinated to occur only after the neurons $g_{i+1}, s_{i+1}$ and $e_{i+1}$ obtain their values. 

The timing chain $T$ consists of $\tau = \log n \cdot \tau'$ neurons $\sigma_1, \ldots \sigma_{\tau}$, where $\tau' = \Theta(\log m + \log \log n)$ is an upper bound on the computation time of the Count-Min sub-networks. 
The first neuron $\sigma_1$ has an incoming edge from the input neuron $a$ with weight $1$ and bias $1$. For $i=2, \ldots ,\tau$, the neuron $\sigma_i$ has an incoming edge from $\sigma_{i-1}$ with weight $1$ and bias $1$. The last neuron $\sigma_{\tau}$ is an inhibitory neuron, with outgoing edges to the neurons $\overbar{z}_1, \ldots , \overbar{z}_{\log n}$ and $\overbar{z}'_{1}, \ldots , \overbar{z}'_{\log n}$ with large negative weights. The inhibition of these neurons inhibits their self loops in preparation for the next input.

\item [(4)]
In the high level, for $i \in \{\log n, \ldots, 1\}$, every two consecutive sub-networks $\cC_{i+1}$ and $\cC_{i}$ are connected in a way that guarantees the following. Let  $\chi_{i+1}$ be median candidate at phase $i+1$ of the binary search (i.e., that was fed as input to $\cC_{i+1}$). Let $\freq([x,y])$ be the estimated frequency of the range $[x,y]$ obtained by the Count-Min sketch networks $\cC_{\log n}, \ldots, \cC_{i+1}$. Then candidate $\chi_{i}$ is defined as:
\[
  \chi_{i}= 
  \begin{cases}
    \chi_{i+1} - 2^{i-1}, & \text{if~} \freq([1,\chi_{i+1}]) >  m'/2+ \epsilon/2 m'\\
    \chi_{i+1} + 2^{i-1}, & \text{if~} \freq([1,\chi_{i+1}]) <  m'/2- \epsilon/2 m'~.
  \end{cases}  
\]
In the remaining case where $\freq([1,\chi_{i+1}])\in [m'/2 \pm\epsilon/2 m']$, the candidate $\chi_{i+1}$ will be returned as the output result.

In every step $i$, the candidate $\chi_{i}$ will be encoded using the input neurons of $\cC_{i}$, $\cC'_{i}$ denoted as $\overbar{z}_i, \overbar{z}'_{i}$.
For $i= \log n$, the neurons $\overbar{z}_{\log n}$, $\overbar{z}'_{\log n}$ have incoming edges from the neuron $a$ such that $\dec(\overbar{z}_{\log n})= 2^{\log n-1}-1$\footnote{or $2^{\lfloor  \log n \rfloor-1}-1 $ if $n$ is not a power of $2$}. Hence, the first median candidate $\chi_{\log n}$ will be represented in $\overbar{z}_{\log n}$ as the binary vector $(0,0,1, \ldots, 1)$. 
For index $i \in [1, \log n -1]$, toward updating the input neurons $\overbar{z}_i$ with the candidate $\chi_i$, the network contains $\log n$ intermediate neurons $\overbar{t}_{i}= t_{i,1}, \ldots, t_{i, \log n}$ with the following connectivity:
(i) For $j =1, \ldots, i-1$, the neuron $t_{i,j}$ has incoming edges from  $g_{i+1}$ and  $s_{i+1}$ and serves as an $\OR$ gate.
(ii) For $j = i+2 , \ldots \log n$, the neuron $t_{i,j}$ has incoming edges from  $g_{i+1}$, $s_{i+1}$ with weight $1$, an incoming edge from $\overbar{z}_{i+1,j}$ with weight $2$ and bias $3$. Hence, if either $g_{i+1}$ or $s_{i+1}$ fired in round $\tau$, the firing state of the $j^{th}$ neuron $t_{i,j}$ in round $\tau+1$, is equal to the firing state of $\overbar{z}_{i+1,j}$ in round $\tau$.
(iii) The neuron $s_{i+1}$ is connected to the neuron $t_{i,i+1}$ with large positive weight.
%
%

Next, the incoming edges of the neurons $\overbar{z}_i$ are set as follows. 
Every neuron  $z_{i,j}$ has incoming edges from $\overbar{t}_{i}$, and $\sigma_{(\log n - i) \cdot  \tau'}$, where we set the weights and bias such that $z_{i,j}$ fires either due to the signal from $\overbar{x}'$ (in case of insertion) or both neurons $t_{i,j}$ and $\sigma_{(\log n - i) \tau'}$ fired.
The incoming edge from the timing chain $T$ guarantees that when we update $\overbar{z}_i$, the computation of the previous candidates $\chi_{i+1}$ has been completed. 


\item [(5)] Once a median estimation is found, the output neurons $\overbar{y}$ are updated in the following manner. 
For every $i \in \{1,\ldots,\log n\}$ the network contains a vector of $\log n$ intermediate neurons $\overbar{y}'_i$. The neurons $\overbar{y}'_i$ are responsible for updating the output neurons $\overbar{y}$ when the candidate $\chi_i$ is returned as the median estimation. 
Every neuron $y'_{i,j}$ serves as an $\AND$ gate of the equality neuron $e_i$ and $z_{i,j}$. 
We then connect the neurons $\overbar{y}'_1, \ldots \overbar{y}'_{\log n-1}$ to the output neurons $\overbar{y}$, where the $j^{th}$ output neuron $y_{i,j}$ serves as an $\OR$ gate of the $j^{th}$ neurons $y'_{1,j}, \ldots \overbar{y}'_{\log n,j}$.
\end{itemize}

\noindent \textbf{Ensuring the Output is a Stream Element.}
We modify the neuron $s_{i}$ to fire also if $\sum_{j=i}^{\log n} \dec(\overbar{f}_{j}) \leq  \dec(\overbar{o})/2- \epsilon/2 \dec(\overbar{o})$ \emph{and} the inhibitory output neurons of $\cC'_{i}$ are idle. This is done using two intermediate neurons (one for each case).
In addition, we set the equality neuron $e_{i}$ to fire only if both $e_{i,1}$, $e_{i,2}$ fire and at least one of the excitatory output neurons of the sub-network $\cC'_{i}$ fires.

\paragraph{Space and Time Complexity.}
The update of the Neural Counter  $\mathcal{NC}_m$ requires $O(\log m)$ rounds.  
Each one of the $2 \log n$ Count-Min sketch network requires $O(\log \log n + \log m)$ rounds. Hence a median query is computed within $O(\log n(\log \log n + \log m))$ rounds, and an element insertion is complete within $O(\log \log n + \log m)$ rounds. 

Regarding the networks size, the network contains $O(\log n)$ Count-Min sketch networks, with parameters $n,m$ and $\epsilon'=O(\epsilon/\log n)$ , $\delta' = O(\delta/\log n)$. Hence, each of the Count-Min sketch networks requires $O(\log n/\epsilon \cdot \log \log n \cdot \log m \cdot \log (\log n/\delta)+\log ^2 (\log n/\delta))$ neurons. 
The neural counter requires $O(\log m)$ neurons, and the timing chain consists of $O(\log (\log n(\log \log n + \log m)))$ neurons. 
Additionally, we introduce $O(\log ^2 n)$ intermediate auxiliary neurons. 
Thus, the approximate median network contains  $\widetilde{O}(1/\epsilon)$
auxiliary neurons.

\paragraph{Correctness.}
We show the proposed network implements the algorithm of~\cite{cormode2005improved}.
Let $x, a$ be an input presented at round $\tau_0$. 
We begin with considering insertion operations, where $x$ represents a stream element $i \in [1,n]$ and $a=0$. In round $\tau_0+1$, the input to the neural counter $a'$ fires, as well as the inhibitor $r'$. Since $r'$ inhibits $\overbar{x}$, and $a'$, the counter  $\mathcal{NC}_m$ is incremented exactly once.

Due to the incoming edges from the input neurons $\overbar{x}$, at round $\tau_0+1$ the neurons $\overbar{x}'$ fire, representing the binary encoding of $i$. Additionally the neuron $b_1, \ldots, b_{\log n}$, $b'_{1}, \ldots, b'_{\log n}$ which are the input neurons to the networks $\cC_1, \ldots, \cC_{\log n}, \cC'_1, \ldots, \cC'_{\log n}$, representing an $\inc(x)$ operation, fire starting round $\tau_0+1$. 
In round $\tau_0+2$ the input neurons $\overbar{z}_1 ,\ldots \overbar{z}_{\log n}, \overbar{z}'_{1} ,\ldots \overbar{z}_{\log n}$ receives the signals from $\overbar{x}'$ and begin to fire. Due to the self-loops which enable persistence, and the modification to the Count Min networks which inhibits these neurons once the computation is complete, by Theorem~\ref{cor:min-sketch} each sub-network $\cC_{i}$ performs the operation $\countt(\dec(\overbar{z}_i))$ by round $\tau_1 =\tau_0+O(\log m + \log \log n)$.

We now show that every sub-network $\cC_{i}$ maintains an estimation of the frequencies of the dyadic intervals $\mathcal{I}_i$. 
For a stream element $k \in [n]$, due to the incoming edges from $\overbar{x}'$, the neurons $\overbar{z}_j$ encodes the value $\lfloor k/2^j \rfloor$ in round $\tau_0+2$. Thus, for every two stream elements $k_1,k_2$ in the interval $[c \cdot 2^i, c \cdot 2^{i+1}-1]$ for $c\in \mathbb{N}$, the input to the network $\cC_{i}$ is identical and equals $c \cdot 2^i$. 
On the other hand, if $k_1/2^i \neq k_2/2^i$, the input presented to the sub-network $\cC_{i}$ when inserting $k_1$ is different than the input presented when inserting $k_2$.
We conclude that the sub-networks $\cC_1, \ldots, \cC_{\log n}$ implement the neural Count-Min data structure for the dyadic intervals of $[n]$.

We next turn to consider a median query presented at round $\tau_0$. Hence, in round $\tau_0$ the input neurons $\overbar{x}$ are idle and the neuron $a$ fires.
We first note that because  $\mathcal{NC}_m$ is incremented once per stream element, and the persistence of each element is $\Omega(\log m)$, in round $\tau_0$ the outputs of the neural counter $\overbar{o}$ and $\overbar{o}'$ encodes the size of the stream in round $\tau_0$ denoted as $m'$.

For every $i= \log n, \ldots, 1$, let $\chi_i$ be the value encoded in the neurons $\overbar{z}_{i}$ in round $\tau_0+i \cdot \tau'+1$, where $\tau'$ is a parameter which upper bounds the computation time of the Count-Min sub-networks. 
Note that the assignment of the candidates $\chi_{\log n},\chi_{\log n -1} \ldots $ is performed in a sequential manner with time intervals of $\tau'$ rounds due to the incoming edges from the timing chain $T$. We consider the candidate $\chi_i$ encoded in the firing state of $\overbar{z}_{i}$ at round $\tau_0 + (\log n- i + 1) \cdot \tau'+1$  as an iteration of a binary search. 

Due to the incoming edges from $\overbar{z}_{i+1}$ described in Step (4), candidate $\chi_i$ (i.e $\dec(\overbar{z}_{\log n - i})$) defers from $\chi_{i+1}$ in the following manner.

\begin{observation} \label{obs:binary-search}
If $\chi_i \neq 0$, $\chi_{i+1} \neq 0$, in case $g_i$ fires in round $\tau_0+(\log n -i + 1) \cdot \tau'+1$ then $\chi_i=\chi_{i+1}-2^{i-1}$ and in case $s_i$ fires $\chi_i=\chi_{i+1}+2^{i-1}$.
\end{observation}
\begin{proof}
Let $b_{\log n}, \ldots ,b_1$ be the binary representation of the candidate $\chi_{i+1}$, represented in the firing state of $\overbar{z}_{i+1}$.
By the definition of Step (4), the neuron $t_{i+1,i+1}$ is idle in every round, and therefore $b_{i+1}=0$.
Additionally, since $\chi_{i+1} \neq 0$, for every $j < i+1$ the neuron $t_{j,i+1}$ fires starting round $\tau_0+(\log n -i) \cdot \tau'$, and therefore $b_{i+1}=1$. 

If $g_{i+1}$ fired, according to step (4) the firing state of the neurons $\overbar{t}_{i}$ in round $\tau_0+i \cdot \tau'+1$ encode the binary representation $(b_{\log n}, \ldots, b_{i+2}, 0, 0 ,1, \ldots, 1) = (b_{\log n}, \ldots, b_{i+2}, b_{i+1}, 0 , b_{i-1}, \ldots, b_1)$. Since $b_i=1$, it follows that $\chi_{i}= \chi_{i+1} - 2^{i-1}$. 

If $s_{i+1}$ fired, the firing state of the neurons $\overbar{t}_{i}$ in round $\tau_0+i \cdot \tau'+1$ encodes the binary representation $(b_{\log n}, \ldots, b_{i+2}, 1, 0 ,1, \ldots, 1) = (b_{\log n}, \ldots, b_{i+2}, 1, 0 , b_{i-1}, \ldots, b_1)$. Since $b_{i}=1$ and $b_{i+1}=0$ we can conclude that $\chi_{i}= \chi_{i+1}+2^{i}-2^{i-1}=\chi_{i+1}+2^{i-1}$. 
%
%
\end{proof}

We next note that due to the definition of the insertion operation to the networks $\cC_1, \ldots, \cC_{\log n}$ when considering the candidate $\chi_i$ the output neurons of $\cC_i$ encode an estimation of the frequency of the interval $[\chi_{i}-2^{i-1}+1,\chi_{i}]$.
\begin{observation} \label{obs:freq-interval}
	For every $\cC_i$, if $\chi_i \neq 0$ is presented at round $t$, by round $t+\tau'$ the output neuron $\overbar{f}_i$ encode an $(1+\epsilon')$-approximation of the frequencies of the interval $[\chi_{i}-2^{i-1}+1,\chi_{i}]$ with probability $1- \delta'$.
\end{observation}
\begin{proof}
By the definition of the step (4) for every coordinate $j < i$ when the $i^{th}$ candidate $\chi_i$ is considered, the neuron $z_{i,j}$ fires (equals one). 
Hence the $i-1$ least significant bits in the binary representation of $\chi_i$ are equal to one. 

In addition, when considering insertion operations, by the definition of the incoming edges of $\overbar{z}_i$, for all elements between $\lfloor \chi_i/2^{i-1} \rfloor $ and $\chi_i$ the input to the network $\cC_i$ is identical and equals to $\chi_i$. Hence, by Theorem~\ref{cor:min-sketch} and the choice of $\tau'$, if $\chi_i$ is presented at round $t$, by round $t+\tau'$ the output neuron $\overbar{f}_i$ encode an $(1+\epsilon')$-approximation of the frequencies of the interval $[\lfloor \chi_i/2^{i-1} \rfloor, \chi_i]=[\chi_{i}-2^{i-1}+1,\chi_{i}]$ with probability $1- \delta'$
\end{proof}

Toward proving our search method implements the algorithm of~\cite{cormode2005improved}, we show that as long as a median estimation has not been found,  
for every candidate $\chi_i$ the output neurons $\overbar{f}_{\log n}, \ldots, \overbar{f}_{i}$ encode the frequency of the range $[1, \chi_i]$.
\begin{claim} \label{clm:med-iteration}
If the output neurons of the network $\overbar{y}$ did not fire by round $\tau_0+(\log n - i + 1) \cdot \tau'+4$, in round $\tau_0+(\log n - i + 1) \cdot \tau'$ it holds that
$\dec(\overbar{f}_{\log n})+\dots + \dec(\overbar{f}_{i})$ encodes a $(1+ (\log n - i + 1) \cdot \epsilon')$ approximation for the frequency of $[1,\chi_i]$ with probability $1-(\log n - i + 1) \cdot \delta'$.        
\end{claim}
\begin{proof}
By induction on $i$ starting $i=\log n$ towards $i=1$. For $i= \log n$ in round $\tau_0+1$ the neurons $\overbar{z}_{\log n}$ encode the element $\chi_{\log n} = 2^{\log n -1} -1$, and by Observation~\ref{obs:freq-interval}, the output neurons $\overbar{f}_{\log n}$ encodes an $1+\epsilon'$ approximation of the frequency of the dyadic interval $[1, 2^{\log n -1} -1]$ by round $\tau_0+\tau'$. Assume the claim holds for the $(i+1)^{th}$ candidate, and consider phase $i$. 
	
For every $j$ let $\tau(j)= \tau_0 + (\log n - j) \cdot \tau'$.
Since the output neurons $\overbar{y}$ did not fire by round $\tau(i-1)+4$, no equality neuron $e_{j}$ has fired previously. Thus, $\chi_{i} \neq 0$ in round $\tau(i)$. Moreover, due to the timing chain $T$ starting round $\tau(i+1)+1$ it holds that $\chi_{i+1} = \dec(\overbar{z}_{i + 1}) \neq 0$. 
	
Hence, by the induction assumption it holds that $f_1 =\dec(\overbar{f}_{\log n-1})+\dots + \dec(\overbar{f}_{i+1})$ encodes an $(1+(\log n - i) \cdot \epsilon')$ approximation of the frequency of $[1,\chi_{i+1}]$ by round $\tau(i)$ with probability $1-(\log n - i) \cdot \delta'$.
Since $\chi_i \neq 0$ it holds that in round $\tau(i)+1$ either $g_{i+1}$ fired, or $s_{i+1}$ fired.
\begin{itemize}
\item If $s_{i+1}$ fired, then by Observation~\ref{obs:binary-search} it holds that $\chi_i=\chi_{i+1}+2^{i-1}$. 
When we query $\cC_i$ by Claim~\ref{obs:freq-interval} in round $\tau_0+(\log n - i+1)\tau'$ it holds that $\dec(\overbar{f}_{i})$ holds an $(1+\epsilon')$ approximation for the frequency of $[\chi_i-2^{i-1}+1,\chi_i] = [\chi_{i-1}+1,\chi_i]$ with probability $1- \delta'$. Thus, we conclude that $\dec(\overbar{f}_{\log n})+\dots + \dec(\overbar{f}_{\log n - i}) = f_1 + \dec(\overbar{f}_{\log n -i})$ is a $(1+(\log n-i+1) \epsilon')$-approximation of the frequency $[1,\chi_i]$ with probability $1- (\log n-i+1) \cdot \delta'$.

\item If $g_{\log n -i + 1}$ fired, by Observation~\ref{obs:binary-search} it holds that $\chi_i = \chi_{i+1} - 2^{i-1}$. 
Recall that by step (2), $g'_{i+1}$ inhibits the neurons $\overbar{f}_{i+1}$ starting round $\tau_0+(\log n-i+1)\tau'-1$ \footnote{The parameter $\tau'$ is chosen to be large enough for that purpose}. 
By Observation~\ref{obs:freq-interval} the neurons $\overbar{f}_{i+1}$ holds an approximation of the frequency of of the interval $[\chi_{i+1}-2^{i}+1, \chi_{i+1}$.
Thus, in round $\tau_0+(\log n - i+1)\tau'$ it holds that $\dec(\overbar{f}_{\log n-1})+\dots + \dec(\overbar{f}_{i})= f_1 - \dec(f_{i+1})+ \dec(\overbar{f}_{i})$ is an $(1+i\cdot \epsilon')$-approximation of the frequency of

$[1,\chi_{i+1}-2^{i}]+[\chi_i-2^{i-1}+1,\chi_i] = [1,(\chi_{i}+2^{i-1})-2^{i}]+[\chi_i-2^{i-1}+1,\chi_i]= [1,\chi_{i}]$ with probability $1- (\log n -i+1) \delta$.
\end{itemize}
\end{proof}
Combining Claim~\ref{clm:med-iteration} with Steps (2) and (4) of the network description we conclude that in every iteration $i$, either we find a median estimation due to the equality neuron $e_i$, or our candidate $\chi_i$ is too small and we increase the next candidate by $n/2^i$, or our candidate is too large and we decrease it by $n/2^i$. The proof of Theorem~\ref{lem:approx-median} then follows from the choice of $\epsilon', \delta'$ and Fact~\ref{fct:median}.

\vspace{-10pt}\section{Streaming Lower Bounds Yield SNN Lower Bounds}

We conclude by addressing Question \ref{ques:lower}, giving a generic reduction that lets us simulate a space-efficient SNN with a space-efficient neural network. This establishes a tight connection between the two models -- any streaming space lower bound yields a near-matching neural-space lower bound. 

\paragraph{Complexity classes in the SNN model.}
For integer parameters $n,m,S$, let $\SNNdet(n,m,S)$ be the set of all data-stream problems $P_{n,m}$ defined over universe $[n]$ and stream length at most $m$ that are solvable by a deterministic SNN with (i) at most $O(S)$ non-input neurons (i.e., auxiliary and output neurons) and (ii) polynomially bounded edge weights (by $n$ and $m$).
Let $\PSNNdet(n,m,S)$ be the class of all data-stream problems $P_{n,m}$ in $\SNNdet(n,m,S)$ whose network solution also have in addition a polynomial persistence time (in $n$ and $m$).  That is, the problems in $\PSNNdet(n,m,S)$ are solvable in polynomial-time by a deterministic SNN that has properties (i,ii). 

We also consider the class of data-stream problems that are solvable by a randomized SNN. 
Let $\SNNrand(n,m,S,\delta)$ be the set of all data-stream problems $P_{n,m}$ that are solvable by a randomized SNN with: (i) at most $O(S)$ non-input neurons, (ii) polynomially bounded edge weights, and (iii) $\le \delta$ failure probability on any input.
The class $\PSNNrand(n,m,S,\delta)$ is a sub-class of $\SNNrand(n,m,S,\delta)$ that requires also a polynomial persistence time.

\paragraph{Complexity classes in the streaming model.}
Let $\Streamdet(n,m, S)$ be the class of all data-stream problems for which there exists a single-pass deterministic streaming algorithm for the problem using space $O(S)$ (potentially with exponentially large update time). Also, let $\Streamrand(n,m, S,\delta)$ be the class of all data-stream problems for which there exists a single-pass randomized streaming algorithm that solves the problem with failure probability $\le \delta$ using space $O(S)$. One can also define the classes $\PStreamdet(n,m,S)$ and $\PStreamrand(n,m,S,\delta)$ which require polynomial update time.

We start by showing that any deterministic SNN with space $S$ for a given data-stream problem $P_{n,m}$ yields an $S$-space deterministic streaming algorithm for the problem.
\begin{lemma}
For every $n,m,S$, we have:\\
$\SNNdet(n,m,S) \subseteq \Streamdet(n,m,S)$ and $\PSNNdet(n,m,S) \subseteq \PStreamdet(n,m,S)$.
\end{lemma}
\begin{proof}
Fix the parameters $n,m,S$, and consider a problem $\Pi \in \SNNdet(n,m,S)$. Let $\mathcal{N}$ be the SNN for the problem $\Pi$. Thus $\mathcal{N}$ has $S$ auxiliary and output neurons. We now describe a streaming algorithm for $\Pi$ that uses space $S$. The algorithm traverses the stream and feeds each item as an input to the network $\mathcal{N}$ (with sufficient large persistence time). Importantly, when considering the subsequent input item, the streaming algorithm only keeps the current firing states of the $S$ auxiliary and output neurons. The correctness follows immediately by the correctness of the network $\mathcal{N}$. The space complexity is $S$ bits corresponding to the firing states of the (non-input) neurons in $\mathcal{N}$.
The proof that $\PSNNdet(n,m, S) \subseteq \PStreamdet(n,m, S)$ is analogous since the update time of the streaming algorithm is polynomial in the network size and the persistence time of the network.
\end{proof}


\vspace{-10pt} 
\paragraph{Pseudorandomness for neural networks.}
Our next goal is to simulate space-efficient randomized SNNs for data-stream problems with small-efficient streaming algorithms. The main barrier arises in the case where the edge weights of the network $\cN$ are chosen randomly according to some distribution. Since an $S$-space network with $n$ input neurons might have $\Omega (S n + S^2)$ edges, the explicit specification of the edge weights is too costly for our purposes.

To overcome this barrier, we will use pseudorandom generators~\cite{vadhan2012pseudorandomness}.
\begin{definition}[PRG]\label{def:PRG}
A deterministic function $G:\{0,1\}^d \to \{0,1\}^n$ for $d< n$ is a $(t,\epsilon)$ pseudorandom generator (PRG) if
any circuit $C$ of size at most $t$ distinguishes a uniform random string $U \gets \B^n$ from $G(R)$, where $R \gets \B^{d}$, with advantage at most $\epsilon$. The parameter $d$ is called the \emph{seed length}.
\end{definition}

\begin{proposition}\label{prop:PRG}[Prop. 7.8 in \cite{vadhan2012pseudorandomness}]
For all $n \in \mathbb{N}$ and $\epsilon >0$, there exists a (non-explicit) $(n,\epsilon)$ pseudorandom generator (PRG) $G:\{0,1\}^d \to \{0,1\}^n$ with seed length $d=O(\log  n+\log 1/\epsilon)$. 
\end{proposition}
The existence of the PRG from Prop. \ref{prop:PRG} is shown via the probabilistic method. Such a PRG can be found in a brute-force manner, by iterating over all $n$-size circuits and all functions $G:\{0,1\}^d \to \{0,1\}^n$ in some fixed order. The desired function $G^*$ is the \emph{first} function that fools the family of all $n$-size circuits.

%

Since an SNN with $n$ input neurons, $S$ non-input neurons for $S=\poly(n)$, and polynomial persistence time can be computed in polynomial time (and thus also by a circuit of polynomial size), we have the following:
\begin{lemma}\label{lem:snn-rand-low-space}
Any SNN $\mathcal{N}$ with $n$ input neurons, $S$ non-input neurons for $S=\poly(n)$, and persistence time $\poly(n)$ in an $m$-length stream can be simulated using a total space of $O(S+\log(nm))$. The success guarantee of the simulation is $1-1/\poly(n,m)$. 
\end{lemma}
\begin{proof}
Consider a (centralized, offline) algorithm that given an ordered stream of length $m' \leq m$ of elements in $[1,n]$ evaluates the output of the network $\mathcal{N}$ on that stream. This algorithm can be implemented in time $\poly(n,m)$ and thus there exists a circuit of size $M=\poly(n,m)$ that implements this algorithm. Our goal is to simulate this circuit using a random seed of length $d=O(\log (nm))$ while reducing the success guarantee by an additive term of $1/\poly(n,m)$. To do that, we will use the PRG construction of Prop. \ref{prop:PRG} that given a random seed of size $d$ fools the family of all circuits of size at most $M$ with probability $1-1/\poly(M)$. 

We assume that the PRG function $G$ is hard-coded in the streaming algorithm in the following sense. There is a PRG oracle that given a $d$ length seed $R$ and an index $i$ outputs the $i$'th bit of $G(R)$. We can think of the code of $G$ as simply comprising a look up table, but note that this code is not part of the space complexity of the streaming algorithm, which only includes data written by the algorithm while processing the stream. The seed $R$ must be chosen randomly at the beginning of the stream and then stored, and thus is included in the space complexity. We also note that the evaluation time $G$ (i.e., outputting each bit) might be exponential.

We now describe how to simulate $\cN$ using $O(\log (nm)+S)$ space using this oracle. We store the seed of $O(\log (nm))$ random bits $R$ and the current firing states of all non-input neurons in $\cN$. Then, as we traverse the stream, for every data-item in the stream, the algorithm evaluates the network $\cN$ on that data-item using the PRG oracle in the following manner. The simulation works in a round by round and a neuron by neuron fashion which only stores $O(\log nm)$ bits from $G(R)$ at any given time. To evaluate the firing state of neuron $u$ in layer $i\geq 0$, the total incoming edge weight of $u$ is computed as follows. Let $v_1,\ldots, v_k$ be the incoming neighbors of $u$. The firing states in round $i-1$ are stored explicitly (this is indeed within the space bound $S$). For each $v_j$ that fired in round $i-1$ we look up $O(\log nm)$ entries in $G(R)$ which describe the edge weight $w(v_i,u)$. We note that, without loss of generality we can assume that the edge weights have  precision $1/\poly(n,m)$ and thus can be described with $O(\log nm)$ bits. 
Rounding any edge weights to have this precision will not affect the success probability of the network by more than a $1/\poly(n,m)$ factor. We accumulate the edge weights into a value $P$, the total incoming potential of $u$, which again requires $O(\log nm)$ bits to store. Finally, using $P$ we evaluate the probability that $u$ fires in round $i$. We again can round this probability to $1/\poly(n,m)$ precision, and thus by looking up $O(\log nm)$ entries in $G(R)$ evaluate if $u$ fires in round $i$. We proceed in this way, iterating over all $S$ non-input neurons and storing their states in round $i$, before proceeding to the next round. Overall, our space complexity remains bounded by $O(\log (nm)+S)$.


The success probability of the algorithm overall is decreased by an additive $1/\poly(n,m)$ term, due to the rounding of edge weights and probabilities and the use of pseudorandom rather than truly random bits. 
\end{proof}
Lemma \ref{lem:snn-rand-low-space} implies that any randomized SNN with space $S$ that solves a streaming problem $P_{n,m}$ with probability $1-\delta$ in polynomial time translates into a randomized streaming algorithm for $P_{n,m}$ using space of $S+O(\log (nm))$. We therefore have:
\vspace{-2pt}\begin{theorem}\label{thm:SNNconStream}
$\PSNNrand(n,m,S,\delta) \subseteq \Streamrand(n,m,S+O(\log (nm)),\delta+1/\poly(n,m))~.$
\end{theorem}
\noindent A useful implication of Theorem \ref{thm:SNNconStream} is that any space lower bound in the streaming model immediately translates into space lower bound for networks that have a polynomial persistence time on the input stream. 
\vspace{-5pt}
\begin{corollary}\label{cor:lower-bound}
Let $P_{n,m}$ be a data-stream problem for which any randomized streaming algorithm that solves the problem with probability $1-\delta$ requires space $\Omega(S(n,m,\delta))$. Then, any SNN for solving $P_{n,m}$ within polynomial number of rounds with probability at least $1-\delta+1/\poly(n,m)$ requires space of $\Omega(S(n,m,\delta)-\log (nm))$.
\end{corollary}
\begin{proof}
Assume towards contradiction that there is an SNN for solving $P_{n,m}$ with probability at least $1-\delta-\poly(1/m)$ within polynomial number of rounds, and using space of $o(S(n)+\log m)$. Thus, $P_{n,m} \in \PSNNrand(n,m,o(S(n,\delta)-\log m),1-\delta+\poly(1/m))$.
The exact specification of the $\poly(\cdot)$ terms are given by Theorem \ref{thm:SNNconStream}.
By Theorem \ref{thm:SNNconStream}, it then holds that $P_{n,m} \in  \Streamrand(n,m,o(S(n,\delta)),1-\delta)$. Contradiction for the fact that solving $P_{n,m}$ with probability $1-\delta$ requires streaming space of $\Omega(S(n,\delta))$. The corollary follows.
\end{proof}

\paragraph{Acknowledgments.} 
We are very grateful to Eylon Yogev for various discussions on pseudorandom generators. We also thank David Woodruff for helpful discussions on streaming lower bounds.
\newpage
\bibliography{streaming}

\newcommand{\etalchar}[1]{$^{#1}$}
\begin{thebibliography}{KMM{\etalchar{+}}20}

\bibitem[AMS99]{alon1999space}
Noga Alon, Yossi Matias, and Mario Szegedy.
\newblock The space complexity of approximating the frequency moments.
\newblock {\em Journal of Computer and system sciences}, 58(1):137--147, 1999.

\bibitem[BJK{\etalchar{+}}02]{Bar-YossefJKST02}
Ziv Bar{-}Yossef, T.~S. Jayram, Ravi Kumar, D.~Sivakumar, and Luca Trevisan.
\newblock Counting distinct elements in a data stream.
\newblock In {\em Randomization and Approximation Techniques, 6th International
  Workshop, {RANDOM} 2002, Cambridge, MA, USA, September 13-15, 2002,
  Proceedings}, pages 1--10, 2002.

\bibitem[Bla18]{Blasiok18}
Jaroslaw Blasiok.
\newblock Optimal streaming and tracking distinct elements with high
  probability.
\newblock In {\em Proceedings of the Twenty-Ninth Annual {ACM-SIAM} Symposium
  on Discrete Algorithms, {SODA} 2018, New Orleans, LA, USA, January 7-10,
  2018}, pages 2432--2448, 2018.

\bibitem[CCFC02]{charikar2002finding}
Moses Charikar, Kevin Chen, and Martin Farach-Colton.
\newblock Finding frequent items in data streams.
\newblock In {\em International Colloquium on Automata, Languages, and
  Programming}, pages 693--703. Springer, 2002.

\bibitem[CCL18]{chou2018algorithmic}
Chi-Ning Chou, Kai-Min Chung, and Chi-Jen Lu.
\newblock On the algorithmic power of spiking neural networks.
\newblock In {\em 10th Innovations in Theoretical Computer Science Conference
  (ITCS 2019)}. Schloss Dagstuhl-Leibniz-Zentrum fuer Informatik, 2018.

\bibitem[CDIM03]{cormode2003comparing}
Graham Cormode, Mayur Datar, Piotr Indyk, and Shanmugavelayutham Muthukrishnan.
\newblock Comparing data streams using hamming norms (how to zero in).
\newblock {\em IEEE Transactions on Knowledge and Data Engineering},
  15(3):529--540, 2003.

\bibitem[CG07]{chassaing2007efficient}
Philippe Chassaing and Lucas Gerin.
\newblock Efficient estimation of the cardinality of large data sets.
\newblock {\em arXiv preprint math/0701347}, 2007.

\bibitem[CM05]{cormode2005improved}
Graham Cormode and Shan Muthukrishnan.
\newblock An improved data stream summary: the count-min sketch and its
  applications.
\newblock {\em Journal of Algorithms}, 55(1):58--75, 2005.

\bibitem[CW79]{carter1979universal}
J~Lawrence Carter and Mark~N Wegman.
\newblock Universal classes of hash functions.
\newblock {\em Journal of computer and system sciences}, 18(2):143--154, 1979.

\bibitem[CZ20]{ChenZ20}
Zhiwei Chen and Aoqian Zhang.
\newblock A survey of approximate quantile computation on large-scale data.
\newblock {\em {IEEE} Access}, 8:34585--34597, 2020.

\bibitem[DF03]{DurandF03}
Marianne Durand and Philippe Flajolet.
\newblock Loglog counting of large cardinalities (extended abstract).
\newblock In {\em Algorithms - {ESA} 2003, 11th Annual European Symposium,
  Budapest, Hungary, September 16-19, 2003, Proceedings}, pages 605--617, 2003.

\bibitem[DSN17]{dasgupta2017neural}
Sanjoy Dasgupta, Charles~F Stevens, and Saket Navlakha.
\newblock A neural algorithm for a fundamental computing problem.
\newblock {\em Science}, 358(6364):793--796, 2017.

\bibitem[FFGM07]{flajolet2007hyperloglog}
Philippe Flajolet, {\'E}ric Fusy, Olivier Gandouet, and Fr{\'e}d{\'e}ric
  Meunier.
\newblock Hyperloglog: the analysis of a near-optimal cardinality estimation
  algorithm.
\newblock 2007.

\bibitem[HLMP20]{HitronLMP20}
Yael Hitron, Nancy~A. Lynch, Cameron Musco, and Merav Parter.
\newblock Random sketching, clustering, and short-term memory in spiking neural
  networks.
\newblock In {\em 11th Innovations in Theoretical Computer Science Conference,
  {ITCS} 2020, January 12-14, 2020, Seattle, Washington, {USA}}, pages
  23:1--23:31, 2020.

\bibitem[HP19]{HitronP19}
Yael Hitron and Merav Parter.
\newblock Counting to ten with two fingers: Compressed counting with spiking
  neurons.
\newblock In {\em 27th Annual European Symposium on Algorithms, {ESA} 2019,
  September 9-11, 2019, Munich/Garching, Germany}, pages 57:1--57:17, 2019.

\bibitem[HPP20]{HitronPP20}
Yael Hitron, Merav Parter, and Gur Perri.
\newblock The computational cost of asynchronous neural communication.
\newblock In {\em 11th Innovations in Theoretical Computer Science Conference,
  {ITCS} 2020, January 12-14, 2020, Seattle, Washington, {USA}}, pages
  48:1--48:47, 2020.

\bibitem[Ind06]{indyk2006stable}
Piotr Indyk.
\newblock Stable distributions, pseudorandom generators, embeddings, and data
  stream computation.
\newblock {\em Journal of the ACM (JACM)}, 53(3):307--323, 2006.

\bibitem[IP11]{indyk2011k}
Piotr Indyk and Eric Price.
\newblock K-median clustering, model-based compressive sensing, and sparse
  recovery for earth mover distance.
\newblock In {\em Proceedings of the forty-third annual ACM symposium on Theory
  of computing}, pages 627--636, 2011.

\bibitem[IW03]{IndykW03}
Piotr Indyk and David~P. Woodruff.
\newblock Tight lower bounds for the distinct elements problem.
\newblock In {\em 44th Symposium on Foundations of Computer Science {(FOCS}
  2003), 11-14 October 2003, Cambridge, MA, USA, Proceedings}, pages 283--288,
  2003.

\bibitem[KLL16]{karnin2016optimal}
Zohar Karnin, Kevin Lang, and Edo Liberty.
\newblock Optimal quantile approximation in streams.
\newblock In {\em 2016 IEEE 57th Annual Symposium on Foundations of Computer
  Science (FOCS)}, pages 71--78. IEEE, 2016.

\bibitem[KMM{\etalchar{+}}20]{kapralov2020fast}
Michael Kapralov, Aida Mousavifar, Cameron Musco, Christopher Musco, Navid
  Nouri, Aaron Sidford, and Jakab Tardos.
\newblock Fast and space efficient spectral sparsification in dynamic streams.
\newblock In {\em Proceedings of the Fourteenth Annual ACM-SIAM Symposium on
  Discrete Algorithms}, pages 1814--1833. SIAM, 2020.

\bibitem[KNW10]{KaneNW10}
Daniel~M. Kane, Jelani Nelson, and David~P. Woodruff.
\newblock An optimal algorithm for the distinct elements problem.
\newblock In {\em Proceedings of the Twenty-Ninth {ACM} {SIGMOD-SIGACT-SIGART}
  Symposium on Principles of Database Systems, {PODS} 2010, June 6-11, 2010,
  Indianapolis, Indiana, {USA}}, pages 41--52, 2010.

\bibitem[KP20]{Kallaugher20}
John Kallaugher and Eric Price.
\newblock Separations and equivalences between turnstile streaming and linear
  sketching.
\newblock In {\em Symposium on Theory of Computing, {STOC} 2020}, 2020.

\bibitem[LDP16]{lee2016training}
Jun~Haeng Lee, Tobi Delbruck, and Michael Pfeiffer.
\newblock Training deep spiking neural networks using backpropagation.
\newblock {\em Frontiers in Neuroscience}, 10:508, 2016.

\bibitem[LMP17a]{lynch2017computational}
Nancy Lynch, Cameron Musco, and Merav Parter.
\newblock Computational tradeoffs in biological neural networks:
  Self-stabilizing winner-take-all networks.
\newblock {\em Innovations in Theoretical Computer Science}, 2017.

\bibitem[LMP17b]{lynch2017spiking}
Nancy Lynch, Cameron Musco, and Merav Parter.
\newblock Spiking neural networks: An algorithmic perspective.
\newblock In {\em 5th Workshop on Biological Distributed Algorithms (BDA
  2017)}, 2017.

\bibitem[LMP17c]{LynchMP17}
Nancy~A. Lynch, Cameron Musco, and Merav Parter.
\newblock Neuro-ram unit with applications to similarity testing and
  compression in spiking neural networks.
\newblock In {\em 31st International Symposium on Distributed Computing, {DISC}
  2017, October 16-20, 2017, Vienna, Austria}, pages 33:1--33:16, 2017.

\bibitem[LMPV18]{Legenstein0PV18}
Robert~A. Legenstein, Wolfgang Maass, Christos~H. Papadimitriou, and Santosh~S.
  Vempala.
\newblock Long term memory and the densest k-subgraph problem.
\newblock In {\em 9th Innovations in Theoretical Computer Science Conference,
  {ITCS} 2018, January 11-14, 2018, Cambridge, MA, {USA}}, pages 57:1--57:15,
  2018.

\bibitem[LNW14]{LiNW14}
Yi~Li, Huy~L. Nguyen, and David~P. Woodruff.
\newblock Turnstile streaming algorithms might as well be linear sketches.
\newblock In {\em Symposium on Theory of Computing, {STOC} 2014, New York, NY,
  USA, May 31 - June 03, 2014}, pages 174--183, 2014.

\bibitem[LW13]{LiW13}
Yi~Li and David~P. Woodruff.
\newblock A tight lower bound for high frequency moment estimation with small
  error.
\newblock In {\em Approximation, Randomization, and Combinatorial Optimization.
  Algorithms and Techniques - 16th International Workshop, {APPROX} 2013, and
  17th International Workshop, {RANDOM} 2013, Berkeley, CA, USA, August 21-23,
  2013. Proceedings}, pages 623--638, 2013.

\bibitem[LW19]{lynch2019integrating}
Nancy Lynch and Mien~Brabeeba Wang.
\newblock Integrating temporal information to spatial information in a neural
  circuit.
\newblock {\em arXiv preprint arXiv:1903.01217}, 2019.

\bibitem[Maa96]{maass1996computational}
Wolfgang Maass.
\newblock On the computational power of noisy spiking neurons.
\newblock In {\em Advances in neural information processing systems}, pages
  211--217, 1996.

\bibitem[Maa97]{maass1997networks}
Wolfgang Maass.
\newblock Networks of spiking neurons: the third generation of neural network
  models.
\newblock {\em Neural networks}, 10(9):1659--1671, 1997.

\bibitem[Maa00]{maass2000computational}
Wolfgang Maass.
\newblock On the computational power of winner-take-all.
\newblock {\em Neural computation}, 12(11):2519--2535, 2000.

\bibitem[MP80]{munro1980selection}
J~Ian Munro and Mike~S Paterson.
\newblock Selection and sorting with limited storage.
\newblock {\em Theoretical computer science}, 12(3):315--323, 1980.

\bibitem[MPVL19]{0001PVL19}
Wolfgang Maass, Christos~H. Papadimitriou, Santosh~S. Vempala, and Robert~A.
  Legenstein.
\newblock Brain computation: {A} computer science perspective.
\newblock In {\em Computing and Software Science - State of the Art and
  Perspectives}, pages 184--199. 2019.

\bibitem[MRL98]{manku1998approximate}
Gurmeet~Singh Manku, Sridhar Rajagopalan, and Bruce~G Lindsay.
\newblock Approximate medians and other quantiles in one pass and with limited
  memory.
\newblock {\em ACM SIGMOD Record}, 27(2):426--435, 1998.

\bibitem[Mut05]{muthukrishnan2005data}
Shanmugavelayutham Muthukrishnan.
\newblock {\em Data streams: Algorithms and applications}.
\newblock Now Publishers Inc, 2005.

\bibitem[PV19]{PapadimitriouV19}
Christos~H. Papadimitriou and Santosh~S. Vempala.
\newblock Random projection in the brain and computation with assemblies of
  neurons.
\newblock In {\em 10th Innovations in Theoretical Computer Science Conference,
  {ITCS} 2019, January 10-12, 2019, San Diego, California, {USA}}, pages
  57:1--57:19, 2019.

\bibitem[SCL19]{su2019spike}
Lili Su, Chia-Jung Chang, and Nancy Lynch.
\newblock Spike-based winner-take-all computation: Fundamental limits and
  order-optimal circuits.
\newblock {\em Neural Computation}, 31(12):2523--2561, 2019.

\bibitem[TGK{\etalchar{+}}19]{tavanaei2019deep}
Amirhossein Tavanaei, Masoud Ghodrati, Saeed~Reza Kheradpisheh, Timoth{\'e}e
  Masquelier, and Anthony Maida.
\newblock Deep learning in spiking neural networks.
\newblock {\em Neural Networks}, 111:47--63, 2019.

\bibitem[V{\etalchar{+}}12]{vadhan2012pseudorandomness}
Salil~P Vadhan et~al.
\newblock {\em Pseudorandomness}, volume~7.
\newblock Now, 2012.

\bibitem[Val17]{Valiant17}
Leslie~G. Valiant.
\newblock Capacity of neural networks for lifelong learning of composable
  tasks.
\newblock In {\em 58th {IEEE} Annual Symposium on Foundations of Computer
  Science, {FOCS} 2017, Berkeley, CA, USA, October 15-17, 2017}, pages
  367--378, 2017.

\bibitem[Woo04a]{woodruff2004optimal}
David Woodruff.
\newblock Optimal space lower bounds for all frequency moments.
\newblock In {\em Proceedings of the fifteenth annual ACM-SIAM symposium on
  Discrete algorithms}, pages 167--175. Society for Industrial and Applied
  Mathematics, 2004.

\bibitem[Woo04b]{Woodruff04}
David~P. Woodruff.
\newblock Optimal space lower bounds for all frequency moments.
\newblock In {\em Proceedings of the Fifteenth Annual {ACM-SIAM} Symposium on
  Discrete Algorithms, {SODA} 2004, New Orleans, Louisiana, USA, January 11-14,
  2004}, pages 167--175, 2004.

\end{thebibliography}
\bibliographystyle{alpha}



\end{document}